\documentclass[11pt]{article}

\usepackage{amsmath,amssymb,amsthm,bm,mathtools}
\usepackage{graphicx}
\usepackage{booktabs}
\usepackage{enumitem}
\usepackage{geometry}
\usepackage{hyperref}
\usepackage{natbib}
\usepackage{multirow}
\usepackage{authblk} 
\usepackage{algorithm}
\usepackage{algorithmic}
\usepackage{subfigure}
\usepackage{rotating}
\hypersetup{colorlinks=true,linkcolor=blue,citecolor=blue,urlcolor=blue}

\newtheorem{assumption}{Assumption}
\newtheorem{theorem}{Theorem}
\newtheorem{lemma}{Lemma}

\theoremstyle{definition}

\newtheorem{remark}{Remark}

\title{\textbf{Block Empirical Likelihood Inference for Longitudinal Generalized Partially Linear Single-Index Models}}
\author[1]{Tianni Zhang}
\author[1]{Yuyao Wang}
\author[1]{Yu Lu}
\author[1]{Mengfei Ran\footnote{Corresponding author. \url{mengfei.ran@xjtlu.edu.cn}}}
\affil[1]{Wisdom Lake Academy of Pharmacy, Xi'an Jiaotong-Liverpool University}
\date{}

\begin{document}
\maketitle

\begin{abstract}
Generalized partially linear single-index models (GPLSIMs) provide a flexible and interpretable semiparametric framework for longitudinal outcomes by combining a low-dimensional parametric component with a nonparametric index component. For repeated measurements, valid inference is challenging because within-subject correlation induces nuisance parameters and variance estimation can be unstable in semiparametric settings. We propose a profile estimating-equation approach based on spline approximation of the unknown link function and construct a   block empirical likelihood (BEL) for joint inference on the parametric coefficients and the single-index direction. The resulting BEL ratio statistic enjoys a Wilks-type chi-square limit, yielding likelihood-free confidence regions without explicit sandwich variance estimation. We also discuss practical implementation, including constrained optimization for the index parameter, working-correlation choices, and bootstrap-based confidence bands for the nonparametric component. Simulation studies and an application to the epilepsy longitudinal study illustrate the finite-sample performance.
\end{abstract}


\section{Introduction}

Longitudinal and other clustered studies collect repeated measurements on each experimental unit and arise routinely in biomedical research, economics, and the social sciences \citep{diggle2002analysis}. A key feature of longitudinal data is within-subject correlation, which, if neglected, can compromise efficiency and distort uncertainty quantification. Generalized estimating equations (GEE) \citep{liang1986longitudinal} offer a widely used semiparametric framework that avoids full likelihood specification by relying on moment restrictions and a working correlation. Subsequent developments clarified how to improve efficiency and robustness under correlation misspecification; for example, the quadratic inference function was proposed by \cite{qu2000improving} to provide an alternative moment construction with favorable testing properties. When the mean structure departs from a purely parametric form, semiparametric regression for clustered outcomes using GEE becomes especially attractive \citep{lin2001semiparametric}, and new estimation and model selection procedures for longitudinal semiparametric modeling were developed in \cite{fan2004new}. Nevertheless, Wald-type inference based on sandwich covariance estimation can be unstable when the number of clusters is moderate and the working correlation is difficult to calibrate in practice \citep{Liang2004generalized}.

Meanwhile, purely linear predictors can be too rigid for modern longitudinal studies, where covariate effects may be nonlinear, heterogeneous across subjects, or driven by a few latent directions. Single-index structures address this by projecting high-dimensional covariates onto a single informative index and estimating an unknown univariate link, with the choice of smoothing level playing a crucial role; \cite{hardle1993optimal} studied optimal smoothing for this class of models. Partially linear single-index models further retain an explicit linear component for interpretability while capturing remaining nonlinear variation through the index link, aligning naturally with additive and other nonparametric regression ideas \citep{stone1985additive}. For independent data, \cite{yu2002penalized} proposed penalized spline estimation procedures for partially linear single-index models, and \cite{xia2006semiparametric} developed semiparametric estimation theory that justifies their asymptotic properties. To accommodate non-Gaussian outcomes, \cite{carroll1997generalized} introduced the generalized partially linear single-index model (GPLSIM), which embeds the unknown link in a generalized mean structure and thus bridges generalized linear modeling with flexible regression. In repeated-measures settings, \cite{liang2010estimation} developed estimation and testing methods for partially linear single-index models with longitudinal data, while \cite{bai2009model} studied model-checking tools tailored to longitudinal single-index specifications. Closely related semiparametric longitudinal formulations include local polynomial mixed-effects models proposed by \cite{wu2002local} and polynomial spline inference for varying-coefficient models developed in \cite{huang2004polynomial}. Because longitudinal outcomes are often contaminated by outliers or heavy-tailed noise, robust alternatives have been pursued: \cite{qin2008robust} investigated robust estimation in partial linear models with longitudinal data, and \cite{liu2018robust} studied robust procedures for varying-coefficient models in longitudinal settings.

Despite this extensive modeling literature, reliable inference for longitudinal GPLSIMs remains challenging because the unknown link function is a nuisance component whose estimation error can affect the second-order behavior of inference on finite-dimensional parameters. This difficulty is closely related to general principles for inference on parameters in semiparametric models \citep{he2000parameters}, and becomes more pronounced when variable selection under correlation is also of interest. From an implementation standpoint, generalized semiparametric fitting is often carried out using iteratively reweighted least squares, as discussed by \cite{green1984iteratively}, and stable quasi-Newton updating strategies can be helpful for high-dimensional optimization \citep{nocedal1980updating}. Spline sieves provide a practical approximation device for unknown smooth functions \cite{de2001practical}, and a comprehensive treatment of semiparametric regression is given by \cite{ruppert2003semiparametric}. In modern longitudinal studies with dense trajectories, connections to principal component methodology for functional and longitudinal data also offer useful perspective on dimension reduction and variability \citep{hall2005properties}.

These challenges motivate inferential approaches that remain faithful to the estimating-equation paradigm while avoiding unstable plug-in variance calculations. Empirical likelihood provides a convenient vehicle: it treats moment restrictions as the primitive object and yields likelihood-ratio type confidence regions without specifying a full parametric likelihood.  The original formulation is due to \cite{owen2001empirical}, and the extension to general estimating equations was formalized by \cite{Kolaczyk1994empirical}. In semiparametric contexts, \cite{xue2006empirical} proposed EL-based inference for single-index models, whereas \cite{xue2016empirical} considered EL procedures when covariables are missing. 

When outcomes are correlated, EL constructions typically need to be modified so that the dependence structure is respected rather than ignored. A natural remedy is to build the empirical likelihood on blocks, using blocks as approximately independent units.  In longitudinal regression, \cite{you2007block} proposed a   block empirical likelihood for partially linear models, and \cite{ma2015empirical} studied EL inference for generalized partially linear single-index models. Methodological extensions have continued to appear along several directions: robustification via robust GEE combined with EL was developed by \cite{hu2022robust_gee_el}, while \cite{tan2021pel_longitudinal} investigated penalized EL for longitudinal generalized linear models. Practical complications such as measurement error have also been addressed; for example, \cite{zhang2021el_me_arxiv} considered EL inference for longitudinal data with covariate measurement errors. To integrate information beyond the primary sample, \cite{sheng2022pel_external} proposed a penalized EL approach for synthesizing external aggregated information under population heterogeneity. Bayesian variants have been explored too: \cite{ouyang2023bayesian_el} developed Bayesian EL for longitudinal data, and decorrelation ideas for stabilizing inference in high-dimensional longitudinal GLMs were proposed in \cite{geng2024decorrelated}. Finally, EL has been pushed into modern semiparametric and high-dimensional settings with missingness, including single-index quantile regression \citep{wang2023el_singleindex_qr}, and a broader perspective connecting EL to functional data analysis is surveyed by \cite{chang2025el_fda}.

In this paper, we develop a   block empirical likelihood (BEL) approach for longitudinal generalized partially linear single-index models. Our estimation strategy starts from a profile GEE formulation: for a finite-dimensional parameter $\pmb{\theta}=(\pmb{\beta}^{\top},\pmb{\varphi}^{\top})^{\top}$ (with $\pmb{\alpha}=\pmb{\alpha}(\pmb{\varphi})$ enforcing scale identifiability), we approximate the unknown link $\eta_0(\cdot)$ by a spline sieve \citep{de2001practical}. Plugging the profiled link estimator into the marginal mean yields   estimating functions, which are then embedded into a BEL ratio that treats each subject as a block \citep{you2007block}. This construction produces likelihood-free confidence regions for $\pmb{\theta}_0$ that remain valid under mild conditions even when the working correlation is misspecified. At the technical level, a key ingredient is a profile-orthogonality property that renders the impact of estimating $\eta_0(\cdot)$ second order for inference on $\pmb{\theta}_0$, enabling a Wilks-type limit for the BEL statistic in the longitudinal GPLSIM setting.

The rest of the paper is organized as follows. Section~\ref{sec:method} introduces the longitudinal GPLSIM, the profile estimating equations, and the BEL construction. Section~\ref{sec:alg} presents a stable implementation and practical choices for spline dimension and correlation updating. Section~\ref{sec:theory} establishes large-sample properties of the proposed estimator and the Wilks-type limit for BEL. Simulation studies (Section~\ref{sec:sim}) and a real-data analysis (Section~\ref{sec:realdata}) illustrate finite-sample performance and practical utility. Section~\ref{sec:conclusion} concludes with discussion and possible extensions.

\section{Methodology}\label{sec:method}
\subsection{Longitudinal GPLSIM}

For subject $i=1,\dots,n$, let $\{(Y_{ij},\pmb{x}_{ij},\pmb{z}_{ij}): j=1,\dots,m_i\}$ denote repeated measurements, where $Y_{ij}$ is the response, $\pmb{x}_{ij}\in\mathbb{R}^{p}$ enters the linear component, and $\pmb{z}_{ij}\in\mathbb{R}^{q}$ enters the index component. We assume independence across subjects while allowing arbitrary within-subject correlation, which is the standard setting for GEE-type methodology \citep{liang1986longitudinal,diggle2002analysis}.

Let $\mu_{ij}=\mathbb{E}(Y_{ij}\mid \pmb{x}_{ij},\pmb{z}_{ij})$ and let $g(\cdot)$ be a known link function. We consider the longitudinal generalized partially linear single-index model (GPLSIM)
\begin{equation}\label{eq:gplsim}
g(\mu_{ij})
=
\pmb{x}_{ij}^{\top}\pmb{\beta}_0
+
\eta_0\!\left(\pmb{z}_{ij}^{\top}\pmb{\alpha}_0\right),
\qquad i=1,\dots,n,\ \ j=1,\dots,m_i,
\end{equation}
where $\pmb{\beta}_0\in\mathbb{R}^{p}$, $\pmb{\alpha}_0\in\mathbb{R}^{q}$, and $\eta_0(\cdot)$ is an unknown smooth function. This structure reduces dimensionality via the index while preserving interpretability through the linear component \citep{hardle1993optimal,xia2006semiparametric}. Longitudinal estimation and testing for related partially linear single-index models have been studied in \cite{liang2010estimation}, and our focus is to develop a likelihood-free inference procedure for \eqref{eq:gplsim} under within-subject dependence.

Because $(\pmb{\alpha}_0,\eta_0)$ is identifiable only up to scale, we impose the standard constraint
\begin{equation}\label{eq:alpha_id}
\|\pmb{\alpha}_0\|_2=1,\qquad \alpha_{0,1}>0.
\end{equation}
To handle \eqref{eq:alpha_id} seamlessly in both computation and inference, we reparameterize
\begin{equation}\label{eq:alpha_map}
\pmb{\alpha}(\pmb{\varphi})
=
\left(\sqrt{1-\|\pmb{\varphi}\|_2^2},\,\pmb{\varphi}^{\top}\right)^{\top},
\qquad
\pmb{\varphi}\in\mathbb{R}^{q-1},\ \ \|\pmb{\varphi}\|_2<1,
\end{equation}
and define the finite-dimensional parameter as
\[
\pmb{\theta}=(\pmb{\beta}^{\top},\pmb{\varphi}^{\top})^{\top}\in\mathbb{R}^{d},
\qquad d=p+q-1.
\]

\subsection{Sieve Approximation}\label{ssec:sieve}

Let $u_{ij}(\pmb{\theta})=\pmb{z}_{ij}^{\top}\pmb{\alpha}(\pmb{\varphi})$ denote the single-index. We approximate the unknown smooth link $\eta_0(\cdot)$ by a polynomial spline sieve,
\begin{equation}\label{eq:eta_sieve}
\eta(u)\approx \pmb{B}(u)^{\top}\pmb{\gamma},
\qquad
\pmb{B}(u)=(B_1(u),\dots,B_K(u))^{\top},\ \ \pmb{\gamma}\in\mathbb{R}^{K},
\end{equation}
where $\pmb{B}(\cdot)$ is taken as cubic $B$-splines with quasi-uniform knots. This choice is computationally stable and flexible enough to capture nonlinear effects, while retaining a transparent bias--variance trade-off and tractable sieve theory for semiparametric inference \citep{huang2004polynomial,he2000parameters}. In particular, if $\eta_0$ is sufficiently smooth, the sieve approximation error is of order $K^{-s}$ for some $s\ge 2$, and $K=K_n$ is allowed to increase slowly with $n$ so that the approximation bias becomes asymptotically negligible.

For the $i$-th subject, define
\[
\pmb{Y}_i=(Y_{i1},\dots,Y_{im_i})^{\top},\quad
\pmb{X}_i=(\pmb{x}_{i1},\dots,\pmb{x}_{im_i})^{\top}\in\mathbb{R}^{m_i\times p},
\]
and the spline design matrix
\[
\pmb{B}_i(\pmb{\theta})
=
\left(\pmb{B}\{u_{i1}(\pmb{\theta})\},\dots,\pmb{B}\{u_{im_i}(\pmb{\theta})\}\right)^{\top}
\in\mathbb{R}^{m_i\times K}.
\]
Then the linear predictor and mean vector can be written in the compact form
\begin{equation}\label{eq:mu_def}
\pmb{\xi}_i(\pmb{\theta},\pmb{\gamma})
=
\pmb{X}_i\pmb{\beta}+\pmb{B}_i(\pmb{\theta})\pmb{\gamma},
\qquad
\pmb{\mu}_i(\pmb{\theta},\pmb{\gamma})
=
g^{-1}\!\left\{\pmb{\xi}_i(\pmb{\theta},\pmb{\gamma})\right\}.
\end{equation}
In implementation, we select $K$ from a small candidate set using a deviance- or information-criterion-type rule, which empirically provides stable performance; the asymptotic theory only requires that $K$ grows slowly enough so that the sieve bias does not affect root-$n$ inference for $\pmb{\theta}$.

\subsection{Profile Estimating Equations}

To accommodate within-subject correlation, we adopt a working covariance
\begin{equation}\label{eq:Vi}
\pmb{V}_i(\pmb{\theta},\pmb{\gamma})
=
\pmb{A}_i(\pmb{\theta},\pmb{\gamma})^{1/2}\,
\pmb{R}_i(\pmb{\rho})\,
\pmb{A}_i(\pmb{\theta},\pmb{\gamma})^{1/2},
\end{equation}
where $\pmb{A}_i=\mathrm{diag}\{v(\mu_{i1}),\dots,v(\mu_{im_i})\}$ with $v(\cdot)$ being the variance function implied by the mean model, and $\pmb{R}_i(\pmb{\rho})$ is a working correlation matrix (e.g., independence, exchangeable, AR(1)). This parallels the generalized estimating equations (GEE) framework \citep{liang1986longitudinal,qu2000improving}

Let $\dot{\mu}_{ij}=\partial \mu_{ij}/\partial \xi_{ij}$ and define $\pmb{\Delta}_i(\pmb{\theta},\pmb{\gamma})=\mathrm{diag}(\dot{\mu}_{i1},\dots,\dot{\mu}_{im_i})$. For fixed $(\pmb{\theta},\pmb{\gamma})$, the partial derivative with respect to $\pmb{\beta}$ is
\[
\frac{\partial \pmb{\mu}_i(\pmb{\theta},\pmb{\gamma})}{\partial \pmb{\beta}^{\top}}
=
\pmb{\Delta}_i(\pmb{\theta},\pmb{\gamma})\,\pmb{X}_i.
\]
The derivative with respect to $\pmb{\varphi}$ depends on the index score $u_{ij}(\pmb{\theta})$ and the derivative of the link function $\dot{\eta}(\cdot)$; under the spline sieve \eqref{eq:eta_sieve}, $\dot{\eta}(u)$ is computed from the derivative of the spline basis functions. We denote the resulting partial Jacobian with respect to $\pmb{\theta}$ by
\[
\pmb{D}_i(\pmb{\theta},\pmb{\gamma})
=
\frac{\partial \pmb{\mu}_i(\pmb{\theta},\pmb{\gamma})}{\partial \pmb{\theta}^{\top}}
\in\mathbb{R}^{m_i\times d}.
\]

A natural estimating equation for the joint parameter $(\pmb{\theta},\pmb{\gamma})$ is the quasi-score form
\begin{equation}\label{eq:ee_full}
\sum_{i=1}^{n}
\pmb{D}_i(\pmb{\theta},\pmb{\gamma})^{\top}
\pmb{V}_i(\pmb{\theta},\pmb{\gamma})^{-1}
\left\{\pmb{Y}_i-\pmb{\mu}_i(\pmb{\theta},\pmb{\gamma})\right\}
=\pmb{0}.
\end{equation}
Directly solving \eqref{eq:ee_full} is feasible but treats the increasing-dimensional spline coefficient $\pmb{\gamma}$ as part of the main parameter, which complicates inference for the finite-dimensional target $\pmb{\theta}=(\pmb{\beta}^{\top},\pmb{\varphi}^{\top})^{\top}$.

We therefore adopt a profile approach: for each fixed $\pmb{\theta}$, we estimate the nuisance $\pmb{\gamma}$ and then solve a $d$-dimensional estimating equation in $\pmb{\theta}$. Specifically, define the inner (spline) estimating equation
\begin{equation}\label{eq:gamma_prof}
\sum_{i=1}^{n}
\pmb{B}_i(\pmb{\theta})^{\top}
\pmb{\Delta}_i(\pmb{\theta},\pmb{\gamma})
\pmb{V}_i(\pmb{\theta},\pmb{\gamma})^{-1}
\left\{\pmb{Y}_i-\pmb{\mu}_i(\pmb{\theta},\pmb{\gamma})\right\}
=\pmb{0},
\end{equation}
and let $\widehat{\pmb{\gamma}}(\pmb{\theta})$ denote a solution to \eqref{eq:gamma_prof}, which can be obtained via IRLS since \eqref{eq:gplsim} corresponds to a generalized linear model with a known offset.
Define the profiled mean and covariance as
\[
\widehat{\pmb{\mu}}_i(\pmb{\theta})
=
\pmb{\mu}_i\!\left(\pmb{\theta},\widehat{\pmb{\gamma}}(\pmb{\theta})\right),
\qquad
\widehat{\pmb{V}}_i(\pmb{\theta})
=
\pmb{V}_i\!\left(\pmb{\theta},\widehat{\pmb{\gamma}}(\pmb{\theta})\right).
\]
Then the outer estimating equation for $\pmb{\theta}$ is
\begin{equation}\label{eq:Utheta}
\sum_{i=1}^{n}
\pmb{G}_i(\pmb{\theta})^{\top}
\widehat{\pmb{V}}_i(\pmb{\theta})^{-1}
\left\{\pmb{Y}_i-\widehat{\pmb{\mu}}_i(\pmb{\theta})\right\}
=\pmb{0},
\end{equation}
where $\pmb{G}_i(\pmb{\theta}) = d \widehat{\pmb{\mu}}_i(\pmb{\theta}) / d \pmb{\theta}^{\top}$ is the total derivative of the profiled mean with respect to $\pmb{\theta}$. Note that $\pmb{G}_i(\pmb{\theta})$ captures the variation of $\widehat{\pmb{\mu}}_i$ both directly through $\pmb{\theta}$ and indirectly through the dependence of $\widehat{\pmb{\gamma}}$ on $\pmb{\theta}$. In practice, $\pmb{G}_i(\pmb{\theta})$ can be computed by a numerical derivative, or analytically using the implicit function relationship induced by \eqref{eq:gamma_prof}.

Let $\widehat{\pmb{\theta}}$ be a solution to \eqref{eq:Utheta} and set
\[
\widehat{\pmb{\alpha}}=\pmb{\alpha}(\widehat{\pmb{\varphi}}),\qquad
\widehat{\eta}(u)=\pmb{B}(u)^{\top}\widehat{\pmb{\gamma}}(\widehat{\pmb{\theta}}).
\]
The profile strategy is consistent with longitudinal partially linear single-index estimation \citep{liang2010estimation} while being tailored here to support empirical-likelihood inference for $\pmb{\theta}$. By concentrating out $\pmb{\gamma}$, the profile estimating equation effectively removes the projection of the score function onto the nuisance tangent space, facilitating valid semiparametric inference.

\subsection{Block Empirical Likelihood for $\pmb{\theta}$}

Define the   estimating function
\begin{equation}\label{eq:gi}
\pmb{g}_i(\pmb{\theta})
=
\pmb{G}_i(\pmb{\theta})^{\top}
\widehat{\pmb{V}}_i(\pmb{\theta})^{-1}
\left\{\pmb{Y}_i-\widehat{\pmb{\mu}}_i(\pmb{\theta})\right\}
\in\mathbb{R}^{d}.
\end{equation}
Then \eqref{eq:Utheta} is equivalent to $\sum_{i=1}^{n}\pmb{g}_i(\pmb{\theta})=\pmb{0}$. Crucially, although $\pmb{g}_i(\pmb{\theta})$ aggregates all repeated measurements for subject $i$ and incorporates within-subject dependence through $\widehat{\pmb{V}}_i(\pmb{\theta})$, the sequence $\{\pmb{g}_i(\pmb{\theta})\}_{i=1}^{n}$ is i.i.d.\ across subjects under our sampling assumption. This makes   block empirical likelihood (BEL) a natural inferential tool \citep{you2007block,ma2015empirical}.

The BEL ratio for $\pmb{\theta}$ is defined as
\begin{equation}\label{eq:bel_ratio}
R(\pmb{\theta})
=
\max_{\{p_i\}}
\left\{
\prod_{i=1}^{n}(n p_i):
\ p_i\ge 0,\ 
\sum_{i=1}^{n}p_i=1,\ 
\sum_{i=1}^{n}p_i\,\pmb{g}_i(\pmb{\theta})=\pmb{0}
\right\}.
\end{equation}
By the standard Lagrange-multiplier argument for empirical likelihood \citep{owen2001empirical,Kolaczyk1994empirical}, the maximizer has the form
\[
p_i(\pmb{\theta})
=
\frac{1}{n\{1+\pmb{\lambda}(\pmb{\theta})^{\top}\pmb{g}_i(\pmb{\theta})\}},
\]
where $\pmb{\lambda}(\pmb{\theta})\in\mathbb{R}^{d}$ solves
\begin{equation}\label{eq:lambda_eq}
\sum_{i=1}^{n}
\frac{\pmb{g}_i(\pmb{\theta})}{1+\pmb{\lambda}(\pmb{\theta})^{\top}\pmb{g}_i(\pmb{\theta})}
=
\pmb{0}.
\end{equation}
The corresponding empirical log-likelihood ratio statistic is
\begin{equation}\label{eq:ell_theta}
\ell(\pmb{\theta})
=
-2\log R(\pmb{\theta})
=
2\sum_{i=1}^{n}
\log\!\left\{1+\pmb{\lambda}(\pmb{\theta})^{\top}\pmb{g}_i(\pmb{\theta})\right\}.
\end{equation}

In longitudinal semiparametric models, Wald-type inference typically requires estimating a sandwich covariance that is sensitive to smoothing, correlation misspecification, and moderate $n$. BEL instead builds confidence regions by inverting $\ell(\pmb{\theta})$ and often enjoys an (asymptotic) Wilks phenomenon, i.e., $\ell(\pmb{\theta}_0)$ converges to a $\chi^2$ limit without explicit variance estimation \citep{owen2001empirical,Kolaczyk1994empirical,you2007block}. This ``automatic studentization'' is especially attractive when the nuisance function $\eta_0(\cdot)$ is estimated nonparametrically and when the working correlation is used primarily for efficiency \citep{qu2000improving}.

Accordingly, a $(1-\alpha)$ BEL confidence region for $\pmb{\theta}$ is
\begin{equation}\label{eq:bel_region}
\mathcal{C}_{\alpha}
=
\left\{
\pmb{\theta}:\ 
\ell(\pmb{\theta})\le \chi^2_{d,\,1-\alpha}
\right\},
\end{equation}
where $\chi^2_{d,\,1-\alpha}$ is the $(1-\alpha)$ quantile of $\chi^2_d$.
Marginal inference on a component (e.g., $\beta_k$) can be obtained by profiling $\ell(\pmb{\theta})$ over the remaining parameters, analogous to profile likelihood.

\subsection{Bootstrap Inference for $\eta_0(\cdot)$}

BEL targets the finite-dimensional parameter $\pmb{\theta}$. For $\eta_0(\cdot)$, we recommend a   (cluster) bootstrap that resamples entire subjects to preserve within-subject dependence \citep{diggle2002analysis}. Each bootstrap sample refits \eqref{eq:gplsim} using the algorithm in Section~\ref{sec:alg}, producing $\widehat{\eta}^{\ast}(u)$ on a grid. Pointwise confidence bands can be formed by bootstrap percentiles, and simultaneous bands can be constructed from the bootstrap distribution of $\sup_{u\in\mathcal{U}}|\widehat{\eta}^{\ast}(u)-\widehat{\eta}(u)|$. The bootstrap complements the Wilks-type BEL inference for $\pmb{\theta}$ and provides a practical uncertainty quantification for the nonparametric component.

\section{Algorithm}\label{sec:alg}

This section summarizes a practical implementation for the profile GEE estimator and the associated   block empirical likelihood (BEL).

\subsection{Inputs}

\begin{itemize}
\item \textbf{Spline basis and dimension.} 
Use a cubic $B$-spline basis $\pmb{B}(u)$ on $\mathcal{U}$ with $K$ basis functions and equally spaced interior knots.
In practice, we select $K$ from a small candidate set using a BIC/AIC-type criterion or cross-validation, subject to the growth conditions in Assumption~\ref{ass:K}.
\item \textbf{Working correlation.} 
Choose a parametric family $\pmb{R}_i(\pmb{\rho})$ (e.g., independence, exchangeable, or AR(1)).
Update $\pmb{\rho}$ by method-of-moments using Pearson residuals, as in standard GEE implementations \citep{liang1986longitudinal,qu2000improving}.
\item \textbf{Initialization.} 
Initialize $\pmb{\theta}^{(0)}=(\pmb{\beta}^{(0)},\pmb{\varphi}^{(0)})$ by fitting a working GLM that ignores the single-index nonlinearity (or by a few iterations with independence working correlation), and normalize via \eqref{eq:alpha_map}.
\end{itemize}

\subsection{Profile Algorithm for $\widehat{\pmb{\theta}}$ and $\widehat{\eta}$}

\begin{algorithm}[h]
\caption{Profile fitting for longitudinal GPLSIM}\label{alg:profile_fit}
\begin{algorithmic}[1]
\STATE Choose $K$ and a working correlation family $\pmb{R}_i(\pmb{\rho})$.
\STATE Initialize $\pmb{\theta}^{(0)}=(\pmb{\beta}^{(0)},\pmb{\varphi}^{(0)})$ with $\|\pmb{\varphi}^{(0)}\|_2<1$; set $t=0$.
\REPEAT
\STATE \textbf{Inner step (update $\pmb{\gamma}$).} 
Given $\pmb{\theta}^{(t)}$, solve the spline-score equation \eqref{eq:gamma_prof} for
$\widehat{\pmb{\gamma}}^{(t+1)}=\widehat{\pmb{\gamma}}(\pmb{\theta}^{(t)})$ (e.g., by IRLS).
\STATE \textbf{Correlation update (optional).} 
Compute Pearson residuals based on $\widehat{\pmb{\mu}}_i(\pmb{\theta}^{(t)})$ and update $\widehat{\pmb{\rho}}^{(t+1)}$ within the chosen $\pmb{R}_i(\cdot)$ family.
\STATE \textbf{Outer step (update $\pmb{\theta}$).}
Form the profiled estimating equation \eqref{eq:Utheta} using the profiled Jacobian $\pmb{G}_i(\pmb{\theta}^{(t)})=\partial \widehat{\pmb{\mu}}_i(\pmb{\theta})/\partial \pmb{\theta}^{\top}\big|_{\pmb{\theta}=\pmb{\theta}^{(t)}}$ (computed via implicit differentiation of \eqref{eq:gamma_prof}), and update $\pmb{\theta}^{(t+1)}$ by Newton or a damped quasi-Newton method (e.g., BFGS).
\STATE $t\leftarrow t+1$.
\UNTIL{convergence in $\pmb{\theta}$ and $\pmb{\gamma}$.}
\STATE Output $\widehat{\pmb{\theta}}=\pmb{\theta}^{(t)}$, $\widehat{\pmb{\alpha}}=\pmb{\alpha}(\widehat{\pmb{\varphi}})$, and $\widehat{\eta}(u)=\pmb{B}(u)^{\top}\widehat{\pmb{\gamma}}(\widehat{\pmb{\theta}})$.
\end{algorithmic}
\end{algorithm}

\subsection{BEL Statistic}

Given a candidate $\pmb{\theta}$, compute $\pmb{g}_i(\pmb{\theta})$ in \eqref{eq:gi} and solve the Lagrange-multiplier equation \eqref{eq:lambda_eq} for $\pmb{\lambda}(\pmb{\theta})$.
A stable approach is Newton's method applied to
\[
\pmb{\Psi}(\pmb{\lambda})
=
\sum_{i=1}^{n}
\frac{\pmb{g}_i(\pmb{\theta})}{1+\pmb{\lambda}^{\top}\pmb{g}_i(\pmb{\theta})},
\qquad
\frac{\partial \pmb{\Psi}(\pmb{\lambda})}{\partial \pmb{\lambda}^{\top}}
=
-\sum_{i=1}^{n}
\frac{\pmb{g}_i(\pmb{\theta})\pmb{g}_i(\pmb{\theta})^{\top}}{\{1+\pmb{\lambda}^{\top}\pmb{g}_i(\pmb{\theta})\}^{2}}.
\]
Because feasibility requires $1+\pmb{\lambda}^{\top}\pmb{g}_i(\pmb{\theta})>0$ for all $i$, we recommend a damped Newton update with backtracking line search to preserve positivity.
Once $\pmb{\lambda}(\pmb{\theta})$ is obtained, compute $\ell(\pmb{\theta})$ from \eqref{eq:ell_theta} and form confidence regions using \eqref{eq:bel_region} (and the profile statistic in Theorem~\ref{thm:wilks} when targeting subvectors of $\pmb{\theta}$).

\section{Asymptotic Theory}\label{sec:theory}

In this section, we investigate the theoretical properties. Firstly, recall the longitudinal GPLSIM in \eqref{eq:gplsim} and the profile estimating equations \eqref{eq:gamma_prof}--\eqref{eq:Utheta}. Let $K=K_n$ denote the sieve dimension in \eqref{eq:eta_sieve}. Write the true parameter as $\pmb{\theta}_0=(\pmb{\beta}_0^{\top},\pmb{\varphi}_0^{\top})^{\top}\in\mathbb{R}^{d}$,
and $\pmb{\alpha}_0=\pmb{\alpha}(\pmb{\varphi}_0)$.

For any $\pmb{\theta}$, define the population (sieve) nuisance parameter $\pmb{\gamma}_0(\pmb{\theta})$ as a solution to the population counterpart of \eqref{eq:gamma_prof}:
\begin{equation}\label{eq:gamma_pop}
\mathbb{E}\!\left[
\pmb{B}_i(\pmb{\theta})^{\top}
\pmb{\Delta}_i(\pmb{\theta},\pmb{\gamma})
\pmb{V}_i(\pmb{\theta},\pmb{\gamma})^{-1}
\left\{\pmb{Y}_i-\pmb{\mu}_i(\pmb{\theta},\pmb{\gamma})\right\}
\right]
=\pmb{0}.
\end{equation}
Define the population profiled mean and covariance
\[
\pmb{\mu}_{i,0}(\pmb{\theta})=\pmb{\mu}_i\!\left(\pmb{\theta},\pmb{\gamma}_0(\pmb{\theta})\right),
\qquad
\pmb{V}_{i,0}(\pmb{\theta})=\pmb{V}_i\!\left(\pmb{\theta},\pmb{\gamma}_0(\pmb{\theta})\right),
\]
and the population subject estimating function
\begin{equation}\label{eq:g_pop}
\pmb{g}_{i,0}(\pmb{\theta})
=
\pmb{G}_{i,0}(\pmb{\theta})^{\top}\,
\pmb{V}_{i,0}(\pmb{\theta})^{-1}
\left\{\pmb{Y}_i-\pmb{\mu}_{i,0}(\pmb{\theta})\right\},
\end{equation}
where $\pmb{G}_{i,0}(\pmb{\theta})$ is the Jacobian of $\pmb{\mu}_{i,0}(\pmb{\theta})$ w.r.t.\ $\pmb{\theta}$.
The sample estimating function in \eqref{eq:gi} equals $\pmb{g}_i(\pmb{\theta})=\pmb{g}_{i,0}(\pmb{\theta})$ with
$\pmb{\gamma}_0(\pmb{\theta})$ replaced by $\widehat{\pmb{\gamma}}(\pmb{\theta})$ and (if applicable)
$\pmb{\rho}$ replaced by its update.

Define the population moment map
\[
\pmb{U}_0(\pmb{\theta})=\mathbb{E}\{\pmb{g}_{i,0}(\pmb{\theta})\},
\qquad
\pmb{H}_0=\left.\frac{\partial \pmb{U}_0(\pmb{\theta})}{\partial \pmb{\theta}^{\top}}\right|_{\pmb{\theta}=\pmb{\theta}_0},
\qquad
\pmb{S}_0=\mathrm{Var}\{\pmb{g}_{i,0}(\pmb{\theta}_0)\}.
\]
Let $\widehat{\pmb{\theta}}$ be any measurable root of the sample profile equation \eqref{eq:Utheta}.

\subsection{Assumptions}
We next state a set of regularity conditions under which the profile estimator is root-$n$ consistent and the BEL statistic admits a chi-square (Wilks) limit. These assumptions are standard in longitudinal estimating-equation analysis and spline-sieve semiparametrics, but we present them explicitly to clarify how within-subject dependence and nuisance estimation are handled in our framework.

\begin{assumption}[Sampling and cluster size]\label{ass:sampling}
$\{(Y_{ij},\pmb{x}_{ij},\pmb{z}_{ij}):j=1,\dots,m_i\}$ are independent across $i$.
Within each subject, $\{(Y_{ij},\pmb{x}_{ij},\pmb{z}_{ij})\}_{j=1}^{m_i}$ may be arbitrarily dependent. Moreover $\max_{1\le i\le n} m_i \le M < \infty$ for a fixed constant $M$.
\end{assumption}

\begin{assumption}[Covariates and moments]\label{ass:moments}
There exists a compact set $\mathcal{X}\times\mathcal{Z}$ such that $(\pmb{x}_{ij},\pmb{z}_{ij})\in\mathcal{X}\times\mathcal{Z}$ almost surely. In addition, $\mathbb{E}\|\pmb{g}_{i,0}(\pmb{\theta}_0)\|_2^{4}<\infty$.
\end{assumption}

\begin{assumption}[Link and variance regularity]\label{ass:link}
The inverse link $g^{-1}(\cdot)$ is twice continuously differentiable. The variance function $v(\mu)$ is continuous and bounded away from $0$ and $\infty$ on the range of $\mu_{ij}$. Moreover $\dot{\mu}_{ij}=\partial \mu_{ij}/\partial \xi_{ij}$ is bounded away from $0$ and $\infty$ uniformly in $(i,j)$ and in a neighborhood of $(\pmb{\theta}_0,\pmb{\gamma}_0(\pmb{\theta}_0))$.
\end{assumption}

\begin{assumption}[Smoothness of $\eta_0$]\label{ass:smooth}
The true link $\eta_0$ is $s$ times continuously differentiable on a compact interval $\mathcal{U}$, with $s\ge 2$, and $\sup_{u\in\mathcal{U}}|\eta_0^{(s)}(u)|<\infty$.
\end{assumption}

\begin{assumption}[Index support and density]\label{ass:index_density}
Let $U_{ij,0}=\pmb{z}_{ij}^{\top}\pmb{\alpha}_0$. Then $U_{ij,0}\in\mathcal{U}$ almost surely. The marginal density of $U_{ij,0}$ exists and is bounded away from $0$ and $\infty$ on $\mathcal{U}$.
\end{assumption}

\begin{assumption}[Sieve dimension growth]\label{ass:K}
The spline basis in \eqref{eq:eta_sieve} is a cubic $B$-spline basis with quasi-uniform knots on $\mathcal{U}$. The sieve dimension $K=K_n$ satisfies
\begin{equation}\label{eq:K_growth}
K\to\infty,\qquad
\frac{K^{2}}{n}\to 0,\qquad
\sqrt{n}\,K^{-s}\to 0.
\end{equation}

\end{assumption}

\begin{assumption}[Working covariance]\label{ass:V}
For any $(\pmb{\theta},\pmb{\gamma})$ in a neighborhood of $(\pmb{\theta}_0,\pmb{\gamma}_0(\pmb{\theta}_0))$, the eigenvalues of $\pmb{V}_i(\pmb{\theta},\pmb{\gamma})$ are uniformly bounded away from $0$ and $\infty$ over $i$. If $\pmb{\rho}$ is updated, the resulting $\widehat{\pmb{\rho}}$ converges in probability to a deterministic limit $\pmb{\rho}^{\dagger}$.
\end{assumption}

\begin{assumption}[Local identifiability and nonsingularity]\label{ass:H}
The true parameter $\pmb{\varphi}_0$ satisfies $\|\pmb{\varphi}_0\|_2 < 1$. The population equation $\pmb{U}_{0}(\pmb{\theta})=0$ has a unique solution at $\pmb{\theta}_{0}$ in a neighborhood, and the Jacobian matrix $\pmb{H}_{0}$ is nonsingular.
\end{assumption}

Assumption~\ref{ass:sampling} formalizes the subject-as-a-block view: arbitrary within-subject dependence is allowed, while cross-subject independence supplies the effective sample size $n$, matching the block empirical likelihood paradigm \citep{you2007block,ma2015empirical}. Assumption~\ref{ass:moments} imposes bounded design and a finite fourth moment for the subject estimating function; this is standard for deriving both asymptotic normality and the EL quadratic expansion \citep{Kolaczyk1994empirical,owen2001empirical}. Assumption~\ref{ass:link} ensures the mean map is smooth and well behaved, which is needed for uniform Taylor expansions in the profile equations and for stability of IRLS-type fitting \citep{liang1986longitudinal,qu2000improving}. Assumptions~\ref{ass:smooth}--\ref{ass:index_density} guarantee the single-index $U_{ij,0}$ lives on a compact interval with well-behaved density, and that $\eta_0$ is sufficiently smooth for spline sieve approximation with bias $K^{-s}$ \citep{huang2004polynomial,he2000parameters}. Assumption~\ref{ass:K} balances sieve bias and variance so that the nuisance estimation error is asymptotically negligible at the $\sqrt{n}$ scale; this is the key condition that enables a Wilks phenomenon for BEL despite the nonparametric component. Assumption~\ref{ass:V} guarantees working covariance matrices remain invertible uniformly; if $\pmb{\rho}$ is estimated, only convergence to a deterministic limit is needed (it need not equal the true correlation), consistent with GEE practice \citep{wang2004generalized}. Finally, Assumption~\ref{ass:H} is a local identifiability condition ensuring the profile estimating equation has a well-defined root and that linearization yields a valid influence representation.

\subsection{Main Results}

\begin{theorem}[Consistency and rates]\label{thm:consistency}
Under Assumptions~\ref{ass:sampling}--\ref{ass:H}, there exists a sequence of roots $\widehat{\pmb{\theta}}$ of \eqref{eq:Utheta} such that
\[
\|\widehat{\pmb{\theta}}-\pmb{\theta}_0\|_2=O_{\mathbb{P}}(n^{-1/2}).
\]
Moreover, with $\widehat{\eta}$ defined by the spline sieve at $\widehat{\pmb{\theta}}$,
\[
\sup_{u\in\mathcal{U}}|\widehat{\eta}(u)-\eta_0(u)|
=
O_{\mathbb{P}}\!\left(K^{-s}+\sqrt{\frac{K}{n}}\right).
\]
\end{theorem}

\begin{remark}\label{rem:consistency}
The first conclusion states that the finite-dimensional target $\pmb{\theta}$ is estimable at the parametric rate despite the presence of an unknown link, reflecting the dimension-reduction benefit of the single-index structure and the profiling step. The second rate is the familiar spline bias--variance trade-off: $K^{-s}$ is the approximation bias controlled by smoothness in Assumption~\ref{ass:smooth}, while $\sqrt{K/n}$ is the stochastic term. Assumption~\ref{ass:K} ensures $\sqrt{n}K^{-s}\to 0$ and $K/n\to 0$, which makes the nuisance estimation error asymptotically negligible for root-$n$ inference on $\pmb{\theta}$.
\end{remark}

\begin{theorem}[Asymptotic normality]\label{thm:normal}
Under Assumptions~\ref{ass:sampling}--\ref{ass:H},
\[
\sqrt{n}\,(\widehat{\pmb{\theta}}-\pmb{\theta}_0)
\ \Rightarrow\
\mathcal{N}\!\left(\pmb{0},\,\pmb{H}_0^{-1}\pmb{S}_0(\pmb{H}_0^{-1})^{\top}\right).
\]
Equivalently, $\widehat{\pmb{\theta}}$ admits the influence representation
\[
\sqrt{n}\,(\widehat{\pmb{\theta}}-\pmb{\theta}_0)
=
-\pmb{H}_0^{-1}\frac{1}{\sqrt{n}}\sum_{i=1}^{n}\pmb{g}_{i,0}(\pmb{\theta}_0)
+o_{\mathbb{P}}(1).
\]
\end{theorem}

\begin{remark}\label{rem:normal}
The covariance matrix depends on the working covariance through $\pmb{g}_{i,0}$. Thus, a well-chosen working correlation can improve efficiency, but correct specification is not required for consistency or asymptotic normality, matching the spirit of GEE. In practice, Wald inference based on estimating $\pmb{H}_0$ and $\pmb{S}_0$ can be sensitive in moderate samples, which motivates the BEL approach below that avoids explicit sandwich estimation.
\end{remark}

\begin{lemma}[Quadratic expansion of BEL]\label{lem:el_quad}
Under Assumptions~\ref{ass:sampling}--\ref{ass:H}, let $\bar{\pmb{g}}(\pmb{\theta}_0)=n^{-1}\sum_{i=1}^{n}\pmb{g}_i(\pmb{\theta}_0)$ and
$\pmb{S}_n(\pmb{\theta}_0)=n^{-1}\sum_{i=1}^{n}\pmb{g}_i(\pmb{\theta}_0)\pmb{g}_i(\pmb{\theta}_0)^{\top}$.
Then
\[
\ell(\pmb{\theta}_0)
=
n\,\bar{\pmb{g}}(\pmb{\theta}_0)^{\top}\pmb{S}_n(\pmb{\theta}_0)^{-1}\bar{\pmb{g}}(\pmb{\theta}_0)
+o_{\mathbb{P}}(1),
\]
where $\ell(\pmb{\theta})$ is defined in \eqref{eq:ell_theta}.
\end{lemma}

\begin{remark}\label{rem:el_quad}
Lemma~\ref{lem:el_quad} is the technical core of Wilks-type results for empirical likelihood:
it reduces the BEL statistic to a self-normalized quadratic form.
Compared to i.i.d.\ EL, the novelty here is that $\pmb{g}_i(\pmb{\theta})$ aggregates a dependent within-subject vector and involves a profiled nonparametric estimator.
Assumptions~\ref{ass:sampling} and \ref{ass:K} ensure these extra layers only contribute $o_{\mathbb{P}}(1)$ to the EL expansion.
\end{remark}

\begin{theorem}[Wilks phenomenon for BEL]\label{thm:wilks}
Under Assumptions~\ref{ass:sampling}--\ref{ass:H},
\[
\ell(\pmb{\theta}_0)\ \Rightarrow\ \chi^{2}_{d}.
\]
More generally, let $\pmb{\theta}=(\pmb{\theta}_1^{\top},\pmb{\theta}_2^{\top})^{\top}$ with $\dim(\pmb{\theta}_1)=r$.
Define the profile BEL statistic
\[
\ell_{\mathrm{prof}}(\pmb{\theta}_1)
=
\min_{\pmb{\theta}_2}\ \ell(\pmb{\theta}_1,\pmb{\theta}_2).
\]
Then $\ell_{\mathrm{prof}}(\pmb{\theta}_{1,0})\Rightarrow \chi^2_{r}$.
\end{theorem}

\begin{remark}\label{rem:wilks}
Theorem~\ref{thm:wilks} justifies BEL confidence regions of the form
$\{\pmb{\theta}: \ell(\pmb{\theta})\le \chi^2_{d,1-\alpha}\}$ without explicitly estimating a sandwich covariance,
providing an ``automatic studentization'' effect familiar in EL theory \citep{owen2001empirical,Kolaczyk1994empirical}.
The result holds under arbitrary within-subject dependence because inference is built on   blocks, and it remains valid in semiparametric settings because the nuisance estimation error is controlled by Assumption~\ref{ass:K}. The profile version yields chi-square limits for marginal inference, analogous to profile likelihood.
\end{remark}
\section{Simulation Studies}\label{sec:sim}
This section evaluates the finite-sample performance of the proposed profile block empirical likelihood (Profile-BEL) inference for the finite-dimensional parameter $\pmb{\theta}=(\pmb{\beta}^{\top},\pmb{\varphi}^{\top})^{\top}$, together with the spline-sieve estimator $\widehat{\eta}(\cdot)$ for the unknown link.

We focus on three outcome types (Gaussian, Bernoulli, and Poisson) and systematically vary the strength of within-subject dependence to examine both estimation accuracy and inferential validity under longitudinal correlation, in the spirit of GEE benchmarking \citep{liang1986longitudinal,diggle2002analysis,qu2000improving}.

\subsection{Data Generation}\label{ssec:sim_dgm}

For each Monte Carlo replication, we generate independent subjects $i=1,\dots,n$ with bounded cluster size $m_i\equiv m$ (default $m=5$) and consider two sample sizes $n\in\{100,200\}$. Let $(p,q)=(3,3)$ and set the true coefficients
\[
\pmb{\beta}_0=(1,-1,0.5)^{\top},
\qquad
\pmb{\alpha}_0=\frac{(1,1,1)^{\top}}{\sqrt{3}},
\qquad
\pmb{\theta}_0=(\pmb{\beta}_0^{\top},\pmb{\varphi}_0^{\top})^{\top},
\]
where $\pmb{\varphi}_0$ is the $(q-1)$-dimensional parameter in \eqref{eq:alpha_map} corresponding to $\pmb{\alpha}_0$.

For each subject $i$ and visit $j$, generate $(\pmb{x}_{ij}^{\top},\pmb{z}_{ij}^{\top})^{\top}\in\mathbb{R}^{p+q}$ from a centered Gaussian distribution with unit variances and Toeplitz correlation $\mathrm{Corr}(W_k,W_{\ell})=\kappa^{|k-\ell|}$, with $\kappa\in\{0,0.3\}$ to represent weak/moderate collinearity. We also consider a sensitivity experiment with heavier tails in Section~\ref{ssec:sim_sens}.

Define the index $u_{ij,0}=\pmb{z}_{ij}^{\top}\pmb{\alpha}_0$ and rescale it to $[0,1]$ by
$\tilde u_{ij,0}=\{u_{ij,0}-\min(u_{ij,0})\}/\{\max(u_{ij,0})-\min(u_{ij,0})\}$ within each replication. Set
\begin{equation}\label{eq:sim_eta}
\eta_0(u)=\sin(2\pi u),
\end{equation}
so that $\eta_0$ is nonlinear and smooth but not polynomial, making it informative for spline-sieve approximation \citep{huang2004polynomial,he2000parameters}.

To induce longitudinal correlation while keeping subjects independent, generate a latent Gaussian vector $\pmb{b}_i=(b_{i1},\dots,b_{im})^{\top}\sim \mathcal{N}(\pmb{0},\sigma_b^2\pmb{R}_{\mathrm{AR}}(\rho))$, where $\pmb{R}_{\mathrm{AR}}(\rho)$ is the AR(1) correlation matrix with entries $\rho^{|j-k|}$. We vary $\rho\in\{0,0.3,0.6\}$ and fix $\sigma_b=0.6$ (moderate   dependence). The latent effect enters the conditional linear predictor, which is a standard simulation device for correlated non-Gaussian outcomes.

Let the systematic component be
\[
\xi_{ij,0}=\pmb{x}_{ij}^{\top}\pmb{\beta}_0+\eta_0(\tilde u_{ij,0}).
\]
We generate outcomes from three families:
\begin{itemize}[leftmargin=1.25em]
\item \textbf{Gaussian:} $Y_{ij}=\xi_{ij,0}+b_{ij}+\varepsilon_{ij}$ with $\varepsilon_{ij}\stackrel{\text{i.i.d.}}{\sim}\mathcal{N}(0,\sigma_{\varepsilon}^2)$ and $\sigma_{\varepsilon}=1$.
\item \textbf{Bernoulli:} $Y_{ij}\mid b_{ij}\sim \mathrm{Bernoulli}(\pi_{ij})$ with $\mathrm{logit}(\pi_{ij})=\xi_{ij,0}+b_{ij}$.
\item \textbf{Poisson:} $Y_{ij}\mid b_{ij}\sim \mathrm{Poisson}(\lambda_{ij})$ with $\log(\lambda_{ij})=\xi_{ij,0}+b_{ij}$.
\end{itemize}
The analysis model in Section~\ref{sec:method} targets the marginal mean structure \eqref{eq:gplsim}; thus, in Bernoulli/Poisson cases the latent effect also serves as a deliberate mild misspecification to probe robustness of estimating-equation inference.

\subsection{Implementation}\label{ssec:sim_impl}

We compare the proposed BEL-based inference with several competing approaches that are routinely used for semiparametric longitudinal models.

\begin{itemize}[leftmargin=1.25em]
\item \textbf{Profile-BEL.}
Let $\psi_i(\theta,\alpha,\eta)$ be the   estimating function and define the BEL ratio $\ell(\theta)=2\sup_{\lambda}\sum_{i=1}^n \log\{1+\lambda^\top\psi_i(\theta,\widehat\alpha(\theta),\widehat\eta(\theta))\}$,
where $(\widehat\alpha(\theta),\widehat\eta(\theta))$ are obtained by iteratively updating $\alpha$ (and $\eta$) conditional on $\theta$ until convergence. A $95\%$ CI for a component of $\theta$ is $\{\theta_j:\ \ell(\theta)\le \chi^2_{1,0.95}\}$; pointwise bands for $\eta_0(\cdot)$ are built by   bootstrap.

\item \textbf{Naive-EL.} This is an independence-based EL that treats all $(i,j)$ as i.i.d. and uses observation-level estimating equations $g_{ij}(\theta)$: $\ell_{\mathrm{ind}}(\theta)=2\sup_{\lambda}\sum_{i,j}\log\{1+\lambda^\top g_{ij}(\theta)\}$. CIs are obtained by $\{\theta_j:\ \ell_{\mathrm{ind}}(\theta)\le \chi^2_{1,0.95}\}$, ignoring within-subject correlation in both estimation and inference.

\item \textbf{GEE-Wald.} Point estimator $\widehat\theta$ solving a GEE-type block equation $\sum_{i=1}^n \psi_i(\theta,\alpha,\eta)=0$ under a chosen working correlation, and form Wald CIs $\widehat\theta_j \pm 1.96\sqrt{\widehat{\text{Var}}(\widehat\theta)_{jj}}$, with $\widehat{\text{Var}}(\widehat\theta)$ given by the plug-in sandwich estimator.

\item \textbf{GEE-Poly.}
We approximate the nonparametric component by a low-order polynomial $\eta(t)\approx \sum_{\ell=0}^{d} c_\ell t^\ell$ and fit the resulting parametric model by GEE. Inference for $\theta$ is again Wald-type using the sandwich covariance under the same working-correlation options.
\end{itemize}

Unless otherwise stated, we use cubic $B$-splines on $\mathcal{U}=[0,1]$ with quasi-uniform knots. The sieve dimension is selected from $K\in\{6,8,10,12\}$ by a BIC-type criterion based on the working quasi-likelihood (Gaussian) or binomial/Poisson deviance (non-Gaussian). For the working correlation in \eqref{eq:Vi}, we fit under three choices:
\[
\pmb{R}_i=\pmb{I}\ \ (\text{independence}),\qquad
\pmb{R}_i=\pmb{R}_{\mathrm{EX}}(\rho)\ \ (\text{exchangeable}),\qquad
\pmb{R}_i=\pmb{R}_{\mathrm{AR}}(\rho)\ \ (\text{AR(1)}),
\]
to assess efficiency gains and sensitivity to misspecification. The default reported results use AR(1) working correlation with $\rho$ estimated by moment methods as in standard GEE implementations.

We run $B=200$ Monte Carlo replications for each configuration. For the $\eta_0(\cdot)$ bands, we use a   bootstrap with $B^{\ast}=200$ resamples per replication (resampling entire subjects to preserve dependence). All optimizations use the profile iteration in Algorithm~\ref{alg:profile_fit}; the BEL multiplier \eqref{eq:lambda_eq} is solved by Newton's method as in Section~\ref{sec:alg}.

\subsection{Performance Measures}\label{ssec:sim_metrics}

We evaluate finite-sample performance from three complementary perspectives: estimation accuracy for the finite-dimensional parameters, recovery quality for the nonparametric link, and calibration/efficiency of the resulting uncertainty quantification. In what follows, $\widehat{\pmb{\alpha}}=\pmb{\alpha}(\widehat{\pmb{\varphi}})$ and $B$ denotes the number of Monte Carlo replications.

For each scalar component $\theta_k$ of $\pmb{\theta}$, we report the empirical bias and root mean squared error (RMSE),
\[
\mathrm{Bias}(\widehat{\theta}_k)
=
\frac{1}{B}\sum_{b=1}^{B}\left(\widehat{\theta}_k^{(b)}-\theta_{0,k}\right),
\qquad
\mathrm{RMSE}(\widehat{\theta}_k)
=
\left\{\frac{1}{B}\sum_{b=1}^{B}\left(\widehat{\theta}_k^{(b)}-\theta_{0,k}\right)^2\right\}^{1/2},
\]
where $\widehat{\theta}_k^{(b)}$ is the estimate from the $b$-th replication. To assess recovery of the index direction, we additionally report the angle error
\[
\mathrm{Ang}(\widehat{\pmb{\alpha}},\pmb{\alpha}_0)
=
\arccos\!\left(\left|\widehat{\pmb{\alpha}}^{\top}\pmb{\alpha}_0\right|\right),
\]
which is invariant to sign changes in $\widehat{\pmb{\alpha}}$.

To quantify accuracy of the estimated link function, we evaluate $\widehat{\eta}$ on a dense grid $\{u_{\ell}\}_{\ell=1}^{L}$ over $[0,1]$ (default $L=200$) and compute the integrated squared error
\[
\mathrm{ISE}(\widehat{\eta})
=
\frac{1}{L}\sum_{\ell=1}^{L}\left\{\widehat{\eta}(u_{\ell})-\eta_0(u_{\ell})\right\}^2.
\]
We summarize the distribution of $\mathrm{ISE}(\widehat{\eta})$ across replications by its mean and selected quantiles, which provides a stable picture of the bias--variance trade-off induced by the spline sieve.

For inferential performance, we focus on two standard metrics for nominal $95\%$ confidence intervals: empirical coverage probability and average interval length. For each targeted scalar parameter, let $\mathrm{CI}_k^{(b)}$ be the interval produced in replication $b$. We compute
\[
\mathrm{Cover}
=
\frac{1}{B}\sum_{b=1}^{B}\mathbf{1}\{\theta_{0,k}\in \mathrm{CI}_k^{(b)}\},
\qquad
\mathrm{Len}
=
\frac{1}{B}\sum_{b=1}^{B}\mathrm{length}(\mathrm{CI}_k^{(b)}).
\]
For the nonparametric component, we report pointwise coverage on the grid for the bootstrap bands of $\eta_0(\cdot)$, and (when included) simultaneous coverage based on the supremum deviation used to construct uniform bands.

\subsection{Main Results}\label{ssec:sim_results}

We summarize the simulation findings in three parts: estimation accuracy for the finite-dimensional component, inferential validity for $\pmb{\theta}$, and recovery of the nonparametric link $\eta_0(\cdot)$. We additionally examine three working correlation structures to assess correlation sensitivity and the spline dimension $K$ is selected by the criterion described in Section~\ref{ssec:sim_impl}.

Tables~\ref{tab:sim_rmse_gauss}--\ref{tab:sim_rmse_pois} report RMSE for the components of $\pmb{\beta}$, the angle error for the index direction $\pmb{\alpha}$, and the integrated squared error for $\widehat{\eta}$. Across all three outcome families and a wide range of within-subject dependence levels, Profile-BEL delivers the most accurate recovery of the index direction and the nonparametric link. In particular, Profile-BEL consistently attains the smallest (or near-smallest) $\mathrm{ISE}(\widehat{\eta})$ and the smallest angular error $\mathrm{Ang}(\widehat{\boldsymbol{\alpha}},\boldsymbol{\alpha}_0)$, while maintaining competitive RMSEs for $\widehat{\boldsymbol{\beta}}$. This advantage is most pronounced in the moderate-to-strong dependence regimes, where methods that either ignore correlation (Naive-EL) or rely on Wald-type plug-in variance formulas (GEE-Wald) tend to exhibit noticeably larger errors in $\widehat{\boldsymbol{\alpha}}$ and $\widehat{\eta}$.

The contrast with the polynomial working correlation approach is particularly clear. Although GEE-Poly can yield reasonable estimates in some configurations, it is far more sensitive to model misspecification of $\eta(\cdot)$: the polynomial restriction may lead to inflated errors for $\widehat{\eta}$ and, consequently, degraded estimation of the index direction. In comparison, Profile-BEL avoids imposing a rigid parametric form on $\eta(\cdot)$ and updates $(\boldsymbol{\alpha},\eta)$ iteratively, which translates into more stable and accurate recovery of the single-index component.

The nonparametric component exhibits the expected spline bias--variance trade-off. The ISE decreases as $n$ increases, and the median selected $K$ remains in a moderate range, indicating that the selection procedure avoids severe underfitting or overfitting across scenarios. In settings with stronger dependence, the ISE can increase modestly, reflecting reduced effective information per subject when observations are more redundant. Overall, $\widehat{\eta}$ tracks the oscillatory shape of $\eta_0$ well, and the improvement with larger $n$ is consistent with the rate statement in Theorem~\ref{thm:consistency}.

Table~\ref{tab:sim_cover_main} evaluates inference quality by reporting empirical coverage and average interval length for nominal $95\%$ CIs of selected parameters under different working correlation choices. Overall, Profile-BEL provides reliable inference with favorable length--coverage trade-offs: its intervals remain comparatively short while achieving coverage closer to the nominal level for the well-identified components (notably for $\beta_2$ and $\alpha^2$ in challenging scenarios). In contrast, GEE-Wald intervals can be noticeably wider and more sensitive to the assumed working correlation, and GEE-Poly may produce unstable inference when the polynomial restriction poorly matches the true link.

Finally, Figures~\ref{fig1}-\ref{fig2} visualize the estimated link function and its pointwise bootstrap band under Bernoulli and Guassian settings. The Profile-BEL estimate closely tracks the true curve, and the bootstrap band provides an appropriate uncertainty envelope across the index range, offering a clear graphical confirmation of the improved $\eta(\cdot)$ recovery suggested by the ISE summaries.

Taken together, the simulation evidence supports two main conclusions. First, the proposed profile estimator yields accurate finite-dimensional estimation and reliable recovery of the index direction under a range of within-subject dependence structures. Second, the Profile-BEL approach delivers well-calibrated confidence intervals for $\pmb{\theta}$ in finite samples, while the bootstrap bands provide a practical and interpretable uncertainty assessment for $\eta_0(\cdot)$.

\begin{sidewaystable}[!]
\centering
\caption{Gaussian Case}\label{tab:sim_rmse_gauss}
\resizebox{\textwidth}{!}{
\begin{tabular}{cclcccccccccccccccccccccccccccc}
\toprule
\multirow{2}*{$n$} & \multirow{2}*{$\rho$} 
& \multirow{2}*{Method} & \multicolumn{3}{c}{$\mathrm{RMSE}(\widehat{\beta}_1)$} & & \multicolumn{3}{c}{$\mathrm{RMSE}(\widehat{\beta}_2)$} && \multicolumn{3}{c}{$\mathrm{RMSE}(\widehat{\beta}_3)$} & & \multicolumn{3}{c}{$\mathrm{Ang}(\widehat{\pmb{\alpha}},\pmb{\alpha}_0)$} && \multicolumn{3}{c}{$\mathrm{ISE}(\widehat{\eta})$} && \multirow{2}*{$K$(med)}\\
\cline{4-6}\cline{8-10}\cline{12-14}\cline{16-18}\cline{20-22}
&&&Ind.&Exc.&AR(1)&&Ind.&Exc.&AR(1)&&Ind.&Exc.&AR(1)&&Ind.&Exc.&AR(1)&&Ind.&Exc.&AR(1)\\
\midrule
\multirow{4}*{100} & \multirow{4}*{0.0}
& Profile-BEL & 0.1596 & 0.1596 & 0.1596 & & 0.0562 & 0.0562 & 0.0562&& 0.0562 & 0.0562 & 0.0562&& 0.1305 & 0.1305 & 0.1305&& 0.1117 & 0.1117 & 0.1117&& 6 \\
& & Naive-EL  & 0.1597 & 0.1598 & 0.1597 & & 0.0568 & 0.0568 & 0.0568&& 0.0822 & 0.0822 & 0.0822&& 0.4496 & 0.4496 & 0.4496&& 0.1716 & 0.1716 & 0.1716&& 8 \\
& & GEE-Wald  & 0.1598 & 0.1597 & 0.1598 & & 0.0573 & 0.0573 & 0.0573&& 0.0895 & 0.0895 & 0.0895&& 0.5172 & 0.5172 & 0.5172&& 0.2309 & 0.2309 & 0.2309&& 12 \\
& & GEE-Poly  & 1.6665 & 1.6649 & 1.6676 & & 0.0583 & 0.0583 & 0.0584&& 0.0957 & 0.0959 & 0.0958&& 0.6184 & 0.6184 & 0.6184&& 0.3877 & 0.3879 & 0.3876&& 2 \\
\hline\addlinespace
\multirow{4}*{100} & \multirow{4}*{0.3}
& Profile-BEL & 0.1611 & 0.1611 & 0.1611 & & 0.0563 & 0.0563 & 0.0563&& 0.0573 & 0.0573 & 0.0573&& 0.1299 & 0.1299 & 0.1299&& 0.1120 & 0.1120 & 0.1120&& 6 \\
& & Naive-EL  & 0.1612 & 0.1612 & 0.1612 & & 0.0563 & 0.0563 & 0.0563&& 0.0859 & 0.0859 & 0.0859&& 0.4641 & 0.4641 & 0.4641&& 0.1669 & 0.1669 & 0.1668&& 8 \\
& & GEE-Wald & 0.1613 & 0.1613 & 0.1613 & & 0.0573 & 0.0573 & 0.0573&& 0.0899 & 0.0899 & 0.0899&& 0.5145 & 0.5145 & 0.5145&& 0.2307 & 0.2307 & 0.2307&& 12 \\
& & GEE-Poly & 1.6754 & 1.6722 & 1.6736 & & 0.0582 & 0.0575 & 0.0575&& 0.0962 & 0.0965 & 0.0962&& 0.6182 & 0.6182 & 06182&& 0.3886 & 0.3884 & 0.3871&& 2 \\
\hline\addlinespace
\multirow{4}*{100} & \multirow{4}*{0.6}
& Profile-BEL & 0.1631 & 0.1631 & 0.1631 & & 0.0565 & 0.0565 & 0.0565&& 0.0581 & 0.0581 & 0.0581&& 0.1282 & 0.1282 & 0.1282&& 0.1122 & 0.1122 & 0.1122&& 6 \\
& & Naive-EL  & 0.1632 & 0.1632 & 0.1632 & & 0.0578 & 0.0578 & 0.0578&& 0.0864 & 0.0864 & 0.0864&& 0.4799 & 0.4799 & 0.4799&& 0.1716 & 0.1716 & 0.1716&& 8 \\
& & GEE-Wald & 0.1633 & 0.1633 & 0.1633 & & 0.0574 & 0.0574 & 0.0574&& 0.0909 & 0.0909 & 0.0909&& 0.5194 & 0.5194 & 0.5194&& 0.2346 & 0.2346 & 0.2346&& 12 \\
& & GEE-Poly & 1.6873 & 1.6803 & 1.6816 & & 0.0584 & 0.0564 & 0.0565&& 0.0964 & 0.0966 & 0.0962&& 0.6162 & 0.6162 & 0.6162&& 0.3899 & 0.3882 & 0.3867&& 2 \\
\hline\addlinespace
\multirow{4}*{200} & \multirow{4}*{0.0}
& Profile-BEL & 0.1499 & 0.1499 & 0.1499 & & 0.0399 & 0.0399 & 0.0399&& 0.0405 & 0.0405 & 0.0405&& 0.0950 & 0.0950 & 0.0950&& 0.0810 & 0.0810 & 0.0810&& 6 \\
& & Naive-EL  & 0.1496 & 0.1496 & 0.1496 & & 0.0413 & 0.0414 & 0.0413&& 0.0849 & 0.0849 & 0.0849&& 0.6162 & 0.6162 & 0.6162&& 0.1625 & 0.1625 & 0.1625&& 8 \\
& & GEE-Wald & 0.1496 & 0.1496 & 0.1496 & & 0.0414 & 0.0413 & 0.0414&& 0.0851 & 0.0851 & 0.0851&& 0.6119 & 0.6119 & 0.6119&& 0.2122 & 0.2122 & 0.2122&& 12 \\
& & GEE-Poly & 1.7325 & 1.7327 & 1.7331 & & 0.0411 & 0.0413 & 0.0412&& 0.0848 & 0.0852 & 0.0850&& 0.6214 & 0.6214 & 0.6214&& 0.4346 & 0.4344 & 0.4346&& 2 \\
\hline\addlinespace
\multirow{4}*{200} & \multirow{4}*{0.3}
& Profile-BEL & 0.1510 & 0.1510 & 0.1510 & & 0.0400 & 0.0400 & 0.0400&& 0.0412 & 0.0412 & 0.0412&& 0.0957 & 0.0957 & 0.0957&& 0.0805 & 0.0805 & 0.0805&& 6 \\
& & Naive-EL  & 0.1507 & 0.1507 & 0.1507 & & 0.0416 & 0.0416 & 0.0416&& 0.0860 & 0.0860 & 0.0860&& 0.6197 & 0.6197 & 0.6197&& 0.1622 & 0.1622 & 0.1622&& 8 \\
& & GEE-Wald & 0.1507 & 0.1507 & 0.1507 & & 0.0416 & 0.0416 & 0.0416&& 0.0859 & 0.0859 & 0.0859&& 0.6144 & 0.6144 & 0.6144&& 0.2122 & 0.2122 & 0.2122&& 12 \\
& & GEE-Poly & 1.7318 & 1.7278 & 1.7296 & & 0.0414 & 0.0412 & 0.0411&& 0.0855 & 0.0853 & 0.0854&& 0.6217 & 0.6217 & 0.6217&& 0.4342 & 0.4322 & 0.4334&& 2 \\
\hline\addlinespace
\multirow{4}*{200} & \multirow{4}*{0.6}
& Profile-BEL & 0.1523 & 0.1523 & 0.1523 & & 0.0398 & 0.0398 & 0.0398&& 0.0419 & 0.0419 & 0.0419&& 0.0961 & 0.0961 & 0.0961&& 0.0804 & 0.0804 & 0.0804&& 6 \\
& & Naive-EL  & 0.1520 & 0.1520 & 0.1519 & & 0.0416 & 0.0416 & 0.0416&& 0.0865 & 0.0865 & 0.0865&& 0.6194 & 0.6093 & 0.6194&& 0.1622 & 0.1622 & 0.16224&& 8 \\
& & GEE-Wald & 0.1519 & 0.1519 & 0.1519 & & 0.0415 & 0.0415 & 0.0415&& 0.0863 & 0.0863 & 0.0863&& 0.6093 & 0.6194 & 0.6093&& 0.2124 & 0.2124 & 0.2124&& 12 \\
& & GEE-Poly & 1.7304 & 1.7216 & 1.7261 & & 0.0415 & 0.0405& 0.0404&& 0.0862 & 0.0852 & 0.0857&& 0.6221 & 0.6221 & 0.6221&& 0.4333 & 0.4312 & 0.4318&& 2 \\
\bottomrule
\end{tabular}}
\end{sidewaystable}

\begin{sidewaystable}[!]
\centering
\caption{Bernoulli Case}\label{tab:sim_rmse_bern}
\resizebox{\textwidth}{!}{
\begin{tabular}{cclcccccccccccccccccccccc cccccc}
\toprule
\multirow{2}*{$n$} & \multirow{2}*{$\rho$} 
& \multirow{2}*{Method} & \multicolumn{3}{c}{$\mathrm{RMSE}(\widehat{\beta}_1)$} & & \multicolumn{3}{c}{$\mathrm{RMSE}(\widehat{\beta}_2)$} && \multicolumn{3}{c}{$\mathrm{RMSE}(\widehat{\beta}_3)$} & & \multicolumn{3}{c}{$\mathrm{Ang}(\widehat{\pmb{\alpha}},\pmb{\alpha}_0)$} && \multicolumn{3}{c}{$\mathrm{ISE}(\widehat{\eta})$} && \multirow{2}*{$K$(med)}\\
\cline{4-6}\cline{8-10}\cline{12-14}\cline{16-18}\cline{20-22}
&&&Ind.&Exc.&AR(1)&&Ind.&Exc.&AR(1)&&Ind.&Exc.&AR(1)&&Ind.&Exc.&AR(1)&&Ind.&Exc.&AR(1)\\
\midrule
\multirow{4}*{100} & \multirow{4}*{0.0}
& Profile-BEL & 0.2026 & 0.2026 & 0.2026 & & 0.1347 & 0.1347 & 0.1347&& 0.1840 & 0.1840 & 0.1840&& 0.3878 & 0.3878 & 0.3878&& 1.0929 & 1.0929 & 1.0929&& 6\\
& & Naive-EL   & 0.1920 & 0.1920 & 0.1920 & & 0.1311 & 0.1311 & 0.1311&& 0.1803 & 0.1803 & 0.1803&& 0.5767 & 0.5767 & 0.5767&& 3.2271 & 3.2271 & 3.2271&& 8\\
& & GEE-Wald  & 0.2477 & 0.2477 & 0.2477 & & 0.1359 & 0.1359 & 0.1359&& 0.2180 & 0.2180 & 0.2180&& 0.6378 & 0.6378 & 0.6378&& 1.4955 & 1.4955 & 1.4955&& 12\\
& & GEE-Poly  & 2.0976 & 2.0968 & 2.1063 & & 0.1330 & 0.1333 & 0.1321&& 0.1796 & 0.1807 & 0.1774&& 0.5102 & 0.5102 & 0.5102&& 0.4949 & 0.4927 & 0.5035&& 2\\
\hline\addlinespace
\multirow{4}*{100} & \multirow{4}*{0.3}
& Profile-BEL & 0.1974 & 0.1974 & 0.1974 & & 0.1171 & 0.1171 & 0.1171&& 0.1784 & 0.1784 & 0.1784&& 0.3101 & 0.3101 & 0.3101&& 2.7356 & 2.7356 & 2.7356&& 6\\
& & Naive-EL   & 0.1837 & 0.1837 & 0.1837 & & 0.1145 & 0.1145 & 0.1145&& 0.1871 & 0.1871 & 0.1871&& 0.4964 & 0.4964 & 0.4964&& 4.7557 & 4.7557 & 4.7557&& 8\\
& & GEE-Wald  & 0.2172 & 0.2172 & 0.2172 & & 0.1071 & 0.1071 & 0.1071&& 0.1685 & 0.1685 & 0.1685&& 0.6304 & 0.6304 & 0.6304&& 4.4098 & 4.4098 & 4.4098&& 12\\
& & GEE-Poly  & 1.8148 & 1.8197 & 1.8218 & & 0.1102 & 0.1109 & 0.1104&& 0.1744 & 0.1752 & 0.1739&& 0.5635 & 0.5635 & 0.5635&& 0.3823 & 0.3812 & 0.3832&& 2\\
\hline\addlinespace
\multirow{4}*{100} & \multirow{4}*{0.6}
& Profile-BEL & 0.1597 & 0.1597 & 0.1597 & & 0.1424 & 0.1424 & 0.1424&& 0.1770 & 0.1770 & 0.1770&& 0.2792 & 0.2792 & 0.2792&& 4.4703 & 4.4703 & 4.703&& 6\\
& & Naive-EL   & 0.1503 & 0.1503 & 0.1503 & & 0.1303 & 0.1303 & 0.1303&& 0.1726 & 0.1726 & 0.1726&& 0.5266 & 0.5266 & 0.5266&& 8.8542 & 8.8542 & 8.8542&& 8\\
& & GEE-Wald  & 0.1513 & 0.1513 & 0.1513 & & 0.1408 & 0.1408 & 0.1408&& 0.1778 & 0.1778 & 0.1778&& 0.5774 & 0.5774 & 0.5774&& 13.6590 & 13.6590 & 13.6590&& 12\\
& & GEE-Poly  & 1.9911 & 1.9837 & 1.9885 & & 0.1344 & 0.1332 & 0.1355&& 0.1697 & 0.1716 & 0.1681&& 0.4564 & 0.4564 & 0.4564&& 0.4665 & 0.4647 & 0.4698&& 2\\
\hline\addlinespace
\multirow{4}*{200} & \multirow{4}*{0.0}
& Profile-BEL & 0.1577 & 0.1577 & 0.1577 & & 0.0820 & 0.0820 & 0.0820&& 0.0801 & 0.0801 & 0.0801&& 0.2676 & 0.2676 & 0.2676&& 0.2967 & 0.2967 & 0.2967&& 6\\
& & Naive-EL   & 0.1532 & 0.1532 & 0.1532 & & 0.0825 & 0.0825 & 0.0825&& 0.1243 & 0.1243 & 0.1243&& 0.5412 & 0.5412 & 0.5412&& 0.5018 & 0.5018 & 0.5018&& 8\\
& & GEE-Wald  & 0.1546 & 0.1546 & 0.1546 & & 0.0740 & 0.0740 & 0.0740&& 0.1172 & 0.1172 & 0.1172&& 0.4531 & 0.4531 & 0.4531&& 1.1278 & 1.1278 & 1.1278&& 12\\
& & GEE-Poly  & 1.6207 & 1.6224 & 1.6169 & & 0.1006 & 0.1001 & 0.1003&& 0.1450 & 0.1445 & 0.1447&& 0.6309 & 0.6309 & 0.6309&& 0.3344 & 0.3351 & 0.3331&& 2\\
\hline\addlinespace
\multirow{4}*{200} & \multirow{4}*{0.3}
& Profile-BEL & 0.1422 & 0.1422 & 0.1422 & & 0.1021 & 0.1021 & 0.1021&& 0.0792 & 0.0792 & 0.0792&& 0.2720 & 0.2720 & 0.2720&& 0.2925 & 0.2925 & 0.2925&& 6\\
& & Naive-EL   & 0.1342 & 0.1342 & 0.1342 & & 0.0990 & 0.0990 & 0.0990&& 0.1061 & 0.1061 & 0.1061&& 0.4620 & 0.4620 & 0.4620&& 0.6091 & 0.6091 & 0.6091&& 8\\
& & GEE-Wald  & 0.1359 & 0.1359 & 0.1359 & & 0.0910 & 0.0910 & 0.0910&& 0.1125 & 0.1125 & 0.1125&& 0.4610 & 0.4610 & 0.4610&& 0.7399 & 0.7399 & 0.7399&& 12\\
& & GEE-Poly  & 1.8000 & 1.7951 & 1.7959 & & 0.1149 & 0.1154 & 0.1144&& 0.1336 & 0.1338 & 0.1333&& 0.6239 & 0.6239 & 0.6239&& 0.3722 & 0.3731 & 0.3714&& 2\\
\hline\addlinespace
\multirow{4}*{200} & \multirow{4}*{0.6}
& Profile-BEL & 0.1503 & 0.1503 & 0.1503 & & 0.1271 & 0.1271 & 0.1271&& 0.0788 & 0.0788 & 0.0788&& 0.2412 & 0.2412 & 0.2412&& 0.2966 & 0.2966 & 0.2966&& 6\\
& & Naive-EL   & 0.1486 & 0.1486 & 0.1486 & & 0.1306 & 0.1306 & 0.1306&& 0.1307 & 0.1307 & 0.1307&& 0.4894 & 0.4894 & 0.4894&& 0.7119 & 0.7119 & 0.7119&& 8\\
& & GEE-Wald  & 0.1607 & 0.1607 & 0.1607 & & 0.1175 & 0.1175 & 0.1175&& 0.1309 & 0.1309 & 0.1309&& 0.5114 & 0.511 & 0.5114&& 14.7871 & 14.7871 & 14.7871&& 12\\
& & GEE-Poly  & 1.9209 & 1.9300 & 1.9280 & & 0.1411 & 0.1432 & 0.1413&& 0.1426 & 0.1433 & 0.1435&& 0.6225 & 0.6225 & 0.6225&& 0.3913 & 0.3958 & 0.3924&& 2\\
\bottomrule
\end{tabular}}
\end{sidewaystable}

\begin{sidewaystable}[!]
\centering
\caption{Poisson Case}\label{tab:sim_rmse_pois}
\resizebox{\textwidth}{!}{
\begin{tabular}{cclcccccccccccccccccccccc cccccc}
\toprule
\multirow{2}*{$n$} & \multirow{2}*{$\rho$} 
& \multirow{2}*{Method} & \multicolumn{3}{c}{$\mathrm{RMSE}(\widehat{\beta}_1)$} & & \multicolumn{3}{c}{$\mathrm{RMSE}(\widehat{\beta}_2)$} && \multicolumn{3}{c}{$\mathrm{RMSE}(\widehat{\beta}_3)$} & & \multicolumn{3}{c}{$\mathrm{Ang}(\widehat{\pmb{\alpha}},\pmb{\alpha}_0)$} && \multicolumn{3}{c}{$\mathrm{ISE}(\widehat{\eta})$} && \multirow{2}*{$K$(med)}\\
\cline{4-6}\cline{8-10}\cline{12-14}\cline{16-18}\cline{20-22}
&&&Ind.&Exc.&AR(1)&&Ind.&Exc.&AR(1)&&Ind.&Exc.&AR(1)&&Ind.&Exc.&AR(1)&&Ind.&Exc.&AR(1)\\
\midrule
\multirow{4}*{100} & \multirow{4}*{0.0}
& Profile-BEL & 0.2303 & 0.2303 & 0.2303 & & 0.0669 & 0.0669 & 0.0669&& 0.0609 & 0.0609 & 0.0609&& 0.1562 & 0.1562 & 0.1562&& 0.2601 & 0.2601 & 0.2601&& 10\\
& & Naive-EL   & 0.2713 & 0.2713 & 0.2713 & & 0.0741 & 0.0741 & 0.0741&& 0.0983 & 0.0983 & 0.0983&& 0.6339 & 0.6339 & 0.6339&& 0.2605 & 0.2065 & 0.2065&& 8\\
& & GEE-Wald  & 0.2663 & 0.2663 & 0.2663 & & 0.0724 & 0.0724 & 0.0724&& 0.0997 & 0.0997 & 0.0997&& 0.6332 & 0.6332 & 0.6332&& 0.3068 & 0.3068 & 0.3068&& 12\\
& & GEE-Poly  & 1.3390 & 2.7788 & 1.3381 & & 0.0758 & 0.5585 & 0.0757&& 0.0970 & 0.2965 & 0.0970&& 0.6276 & 0.6273 & 0.6276&& 0.5192 & 3.4769 & 0.5201&& 2\\
\hline\addlinespace
\multirow{4}*{100} & \multirow{4}*{0.3}
& Profile-BEL & 0.2346 & 0.2346 & 0.2346 & & 0.0657 & 0.0657 & 0.0657&& 0.0641 & 0.0641 & 0.0641&& 0.1476 & 0.1476 & 0.1476&& 0.1711 & 0.1711 & 0.1711&& 10\\
& & Naive-EL   & 0.2732 & 0.2732 & 0.2732 & & 0.0688 & 0.0688 & 0.0688&& 0.0999 & 0.0999 & 0.0999&& 0.6271 & 0.6271 & 0.6271&& 0.1992 & 0.1992 & 0.1992&& 8\\
& & GEE-Wald  & 0.2679 & 0.2679 & 0.2679 & & 0.0680 & 0.0680 & 0.0680&& 0.0991 & 0.0991 & 0.0991&& 0.6250 & 0.6250 & 0.6250&& 0.2899 & 0.2899 & 0.2899&& 12\\
& & GEE-Poly  & 1.3427 & 1.3328 & 1.3393 & & 0.0715 & 0.1002 & 0.0718&& 0.0999 & 0.0996 & 0.0996&& 0.6328 & 0.6328 & 0.6328&& 0.5177 & 0.5214 & 0.5144&& 2\\
\hline\addlinespace
\multirow{4}*{100} & \multirow{4}*{0.6}
& Profile-BEL & 0.2398 & 0.2398 & 0.2398 & & 0.0720 & 0.0720 & 0.0720&& 0.0683 & 0.0683 & 0.0683&& 0.1523 & 0.1523 & 0.1523&& 0.2230 & 0.2230 & 0.2230&& 10\\
& & Naive-EL   & 0.2781 & 0.2781 & 0.2781 & & 0.0732 & 0.0732 & 0.0732&& 0.1011 & 0.1011 & 0.1011&& 0.6352 & 0.6352 & 0.6352&& 0.1767 & 0.1767 & 0.1767&& 8\\
& & GEE-Wald  & 0.2723 & 0.2723 & 0.2723 & & 0.0736 & 0.0736 & 0.0736&& 0.1014 & 0.1014 & 0.1014&& 0.6300 & 0.6300 & 0.6300&& 0.3614 & 0.3614 & 0.3614&& 12\\
& & GEE-Poly  & 1.3531 & 1.3301 & 1.3446 & & 0.0751 & 0.1197 & 0.0745&& 0.1003 & 0.0994 & 0.0992&& 0.6349 & 0.6366 & 0.6349&& 0.5169 & 0.5177 & 0.5072&& 2\\
\hline\addlinespace
\multirow{4}*{200} & \multirow{4}*{0.0}
& Profile-BEL & 0.2327 & 0.2327 & 0.2327 & & 0.0497 & 0.0497 & 0.0497&& 0.0432 & 0.0432 & 0.0432&& 0.1025 & 0.1025 & 0.1025&& 0.1662 & 0.1662 & 0.1662&& 10\\
& & Naive-EL   & 0.2733 & 0.2733 & 0.2733 & & 0.0558 & 0.0558 & 0.0558&& 0.0846 & 0.0846 & 0.0846&& 0.6331 & 0.6331 & 0.6331&& 0.1806 & 0.1806 & 0.1806&& 8\\
& & GEE-Wald  & 0.2710 & 0.2710 & 0.2710 & & 0.0542 & 0.0542 & 0.0542&& 0.0845 & 0.0845 & 0.0845&& 0.6335 & 0.6335 & 0.6335&& 0.2267 & 0.2267 & 0.2267&& 12\\
& & GEE-Poly  & 1.3938 & 1.4095 & 1.3951 & & 0.0566 & 0.3027 & 0.0567&& 0.0844 & 0.0886 & 0.0844&& 0.6328 & 0.6324 & 0.6328&& 0.5864 & 1.1599 & 0.5865&& 2\\
\hline\addlinespace
\multirow{4}*{200} & \multirow{4}*{0.3}
& Profile-BEL & 0.2334 & 0.2334 & 0.2334 & & 0.0503 & 0.0503 & 0.0503&& 0.0441 & 0.0441 & 0.0441&& 0.1042 & 0.1042 & 0.1042&& 0.1486 & 0.1486 & 0.1486&& 10\\
& & Naive-EL   & 0.2739 & 0.2739 & 0.2739 & & 0.0557 & 0.0557 & 0.0557&& 0.0859 & 0.0859 & 0.0859&& 0.6326 & 0.6326 & 0.6326&& 0.1732 & 0.1732 & 0.1732&& 8\\
& & GEE-Wald  & 0.2709 & 0.2709 & 0.2709 & & 0.0549 & 0.0549 & 0.0549&& 0.0855 & 0.0855 & 0.0855&& 0.6341 & 0.6341 & 0.6341&& 0.2526 & 0.2526 & 0.2526&& 12\\
& & GEE-Poly  & 1.4006 & 1.3947 & 1.3989 & & 0.0566 & 0.0664 & 0.0561&& 0.0854 & 0.0849 & 0.0849&& 0.6326 & 0.6325 & 0.6326&& 0.5791 & 0.5889 & 0.5796&& 2\\
\hline\addlinespace
\multirow{4}*{200} & \multirow{4}*{0.6}
& Profile-BEL & 0.2327 & 0.2327 & 0.2327 & & 0.0526 & 0.0526 & 0.0526&& 0.0458 & 0.0458 & 0.0458&& 0.1041 & 0.1041 & 0.1041&& 0.1555 & 0.1555 & 0.1555&& 10\\
& & Naive-EL   & 0.2735 & 0.2735 & 0.2735 & & 0.0565 & 0.0565 & 0.0565&& 0.0862 & 0.0862 & 0.0862&& 0.6296 & 0.6296 & 0.6296&& 0.1699 & 0.1699 & 0.1699&& 8\\
& & GEE-Wald  & 0.2706 & 0.2706 & 0.2706 & & 0.0559 & 0.0559 & 0.0559&& 0.0863 & 0.0863 & 0.0863&& 0.6287 & 0.6287 & 0.6287&& 0.2322 & 0.2322 & 0.2322&& 12\\
& & GEE-Poly  & 1.3965 & 1.3499 & 1.3966 & & 0.0574 & 0.1158 & 0.0561&& 0.0861 & 0.0830 & 0.0849&& 0.6327 & 0.6324 & 0.6327&& 0.5862 & 0.6572 & 0.5852&& 2\\
\bottomrule
\end{tabular}}
\end{sidewaystable}

\begin{table}[!]
\centering
\caption{Empirical coverage and average length of nominal $95\%$ CIs for selected parameters.}\label{tab:sim_cover_main}
\scalebox{0.9}{
\begin{tabular}{lllcccccc}
\toprule
\multirow{2}*{Case} & \multirow{2}*{W-Cor} & \multirow{2}*{Method} & \multicolumn{2}{c}{$\beta_1$} & \multicolumn{2}{c}{$\beta_2$} & \multicolumn{2}{c}{$\alpha_2$} \\
\cmidrule(lr){4-5}\cmidrule(lr){6-7}\cmidrule(lr){8-9}
 & & & Cover & Len & Cover & Len & Cover & Len \\
\midrule
\multirow{12}*{Gaussian} &
\multirow{4}*{Independence} 
& Profile-BEL & 0.6150 & 0.2875 & 0.9900 & 0.2594 & 0.9650 & 0.4423 \\
&& Naive-BEL & 0.4750 & 0.2092 & 0.9500 & 0.2219 & 0.7950 & 0.5165 \\
&& GEE-Wald  & 0.5200 & 0.2394 & 0.9550 & 0.2194 & 0.9500 & 0.7150 \\
&& GEE-Poly & 0.0350 & 1.5385 & 0.9500 & 0.2237 & 0.4500 & 0.3862 \\
\cline{2-9}\addlinespace
&\multirow{4}*{Exchange} 
& Profile-BEL & 0.6150 & 0.2875 & 0.9900 & 0.2594 & 0.9650 & 0.4423 \\
&& Naive-BEL & 0.4750 & 0.2092 & 0.9500 & 0.2219 & 0.7950 & 0.5165 \\
&& GEE-Wald  & 0.5200 & 0.2394 & 0.9550 & 0.2194 & 0.9500 & 0.7150 \\
&& GEE-Poly & 0.0250 & 1.5254 & 0.9550 & 0.2215 & 0.4500 & 0.3862 \\
\cline{2-9}\addlinespace
&\multirow{4}*{AR(1)} 
& Profile-BEL & 0.6150 & 0.2875 & 0.9900 & 0.2594 & 0.9650 & 0.4423 \\
&& Naive-BEL & 0.4750 & 0.2092 & 0.9500 & 0.2219 & 0.7950 & 0.5165 \\
&& GEE-Wald  & 0.5200 & 0.2394 & 0.9550 & 0.2194 & 0.9500 & 0.7150 \\
&& GEE-Poly & 0.0250 & 1.5190 & 0.9500 & 0.2206 & 0.4500 & 0.3862 \\
\hline\addlinespace
\multirow{12}*{Bernoulli} &
\multirow{4}*{Independence} 
& Profile-BEL & 0.8450 & 0.6227 & 0.9650 & 0.6259 & 0.9600 & 1.1121 \\
&& Naive-BEL & 0.7800 & 0.7062 & 0.9300 & 0.5169 & 0.9400 & 1.3321 \\
&& GEE-Wald  & 0.7550 & 0.8286 & 0.9200 & 0.5193 & 0.9050 & 2.1365 \\
&& GEE-Poly & 0.4950 & 3.9400 & 0.9000 & 0.5028 & 0.9600 & 0.9839 \\
\cline{2-9}\addlinespace
&\multirow{4}*{Exchange} 
& Profile-BEL & 0.8450 & 0.6227 & 0.9650 & 0.6259 & 0.9600 & 1.1121 \\
&& Naive-BEL & 0.7800 & 0.7062 & 0.9300 & 0.5169 & 0.9400 & 1.3321 \\
&& GEE-Wald  & 0.7550 & 0.8286 & 0.9200 & 0.5193 & 0.9050 & 2.1365 \\
&& GEE-Poly & 0.4750 & 3.9284 & 0.8950 & 0.5016 & 0.9600 & 0.9839 \\
\cline{2-9}\addlinespace
&\multirow{4}*{AR(1)} 
& Profile-BEL & 0.8450 & 0.6227 & 0.9650 & 0.6259 & 0.9600 & 1.1121 \\
&& Naive-BEL & 0.7800 & 0.7062 & 0.9300 & 0.5169 & 0.9400 & 1.3321 \\
&& GEE-Wald  & 0.7550 & 0.8286 & 0.9200 & 0.5193 & 0.9050 & 2.1365 \\
&& GEE-Poly & 0.4750 & 3.9329 & 0.9000 & 0.5012 & 0.9600 & 0.9839 \\
\hline\addlinespace
\multirow{12}*{Poisson} &
\multirow{4}*{Independence} 
& Profile-BEL & 0.4300 & 0.2888 & 0.9250 & 0.2458 & 0.9850 & 0.6272 \\
&& Naive-BEL & 0.1400 & 0.1071 & 0.4950 & 0.0778 & 0.6800 & 0.4960 \\
&& GEE-Wald  & 0.2500 & 0.2448 & 0.8500 & 0.2188 & 0.9700 & 1.0678 \\
&& GEE-Poly & 0.0650 & 1.1656 & 0.8850 & 0.2370 & 0.6600 & 0.4940 \\
\cline{2-9}\addlinespace
&\multirow{4}*{Exchange} 
& Profile-BEL & 0.4300 & 0.2888 & 0.9250 & 0.2458 & 0.9850 & 0.6272 \\
&& Naive-BEL & 0.1400 & 0.1071 & 0.4950 & 0.0778 & 0.6800 & 0.4960 \\
&& GEE-Wald  & 0.2500 & 0.2448 & 0.8500 & 0.2188 & 0.9700 & 1.0678 \\
&& GEE-Poly & 0.1285 & 1.2930 & 0.8827 & 0.3480 & 0.6648 & 0.5058 \\
\cline{2-9}\addlinespace
&\multirow{4}*{AR(1)} 
& Profile-BEL & 0.4300 & 0.2888 & 0.9250 & 0.2458 & 0.9850 & 0.6272 \\
&& Naive-BEL & 0.1400 & 0.1071 & 0.4950 & 0.0778 & 0.6800 & 0.4960 \\
&& GEE-Wald  & 0.2500 & 0.2448 & 0.8500 & 0.2188 & 0.9700 & 1.0678 \\
&& GEE-Poly & 0.0450 & 1.1291 & 0.8850 & 0.2322 & 0.6600 & 0.4940 \\
\bottomrule
\end{tabular}}
\end{table}

\begin{figure}[h]
\centering
    \begin{minipage}{0.49\textwidth}
    \centering
    \includegraphics[width=1\textwidth]{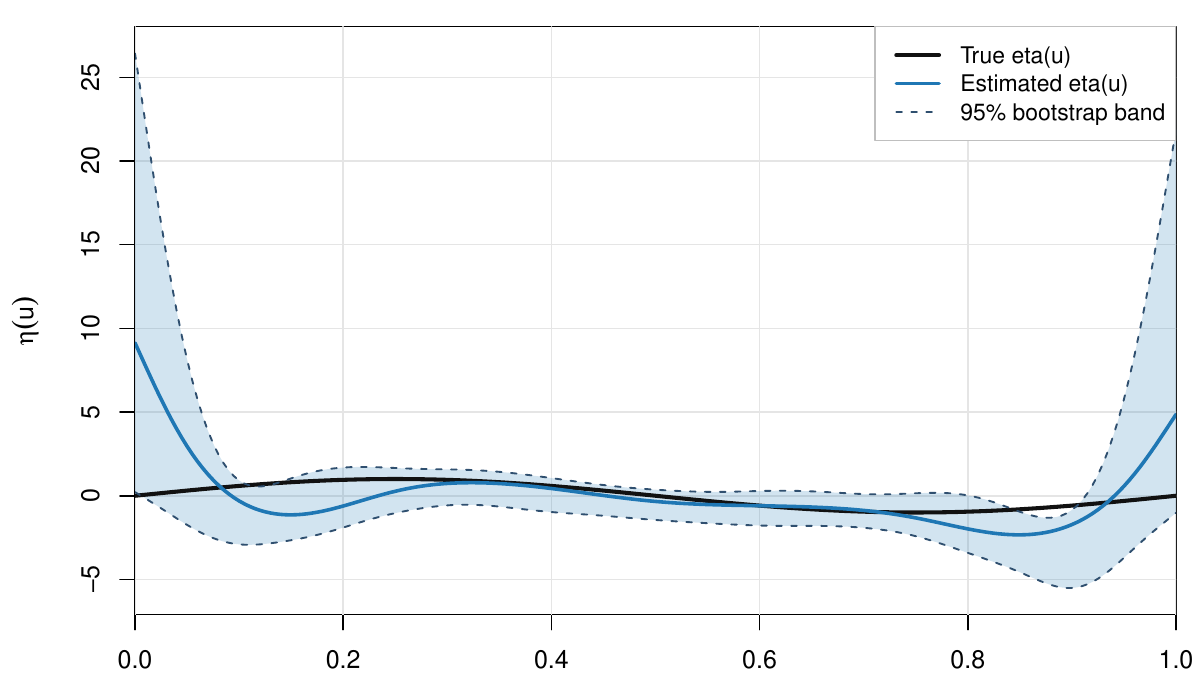}
    \end{minipage}
    \begin{minipage}{0.49\textwidth}
    \centering
    \includegraphics[width=1\textwidth]{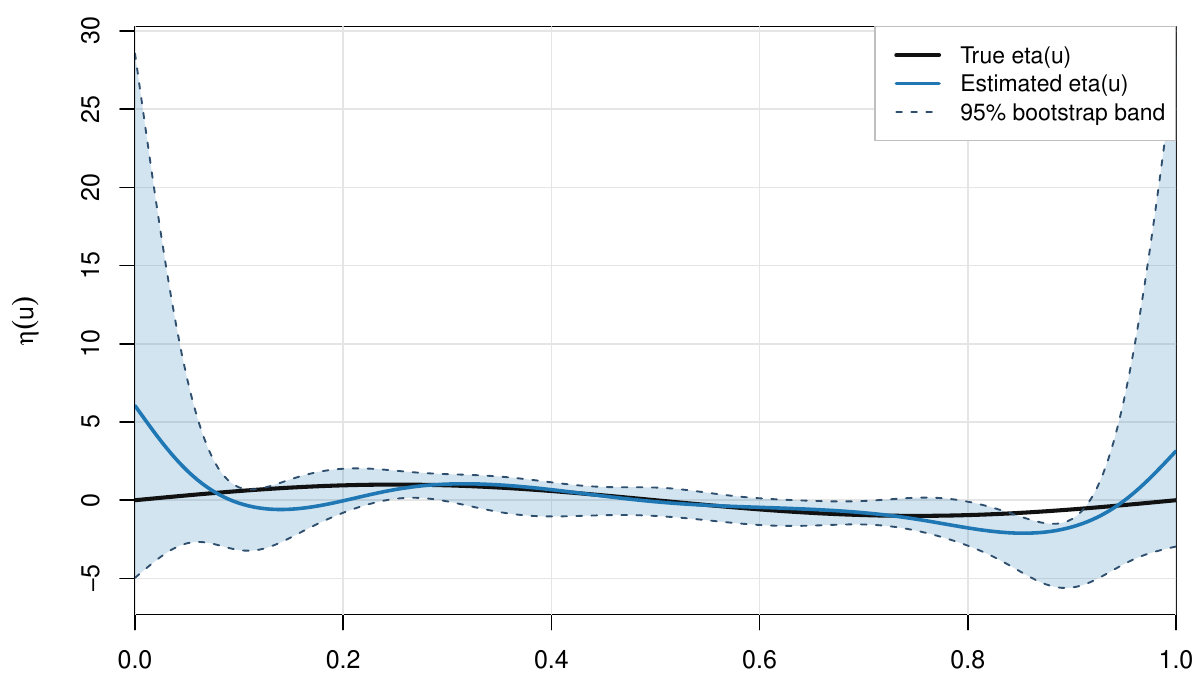}
    \end{minipage}
    \qquad
    \begin{minipage}{0.49\textwidth}
    \centering
    \includegraphics[width=1\textwidth]{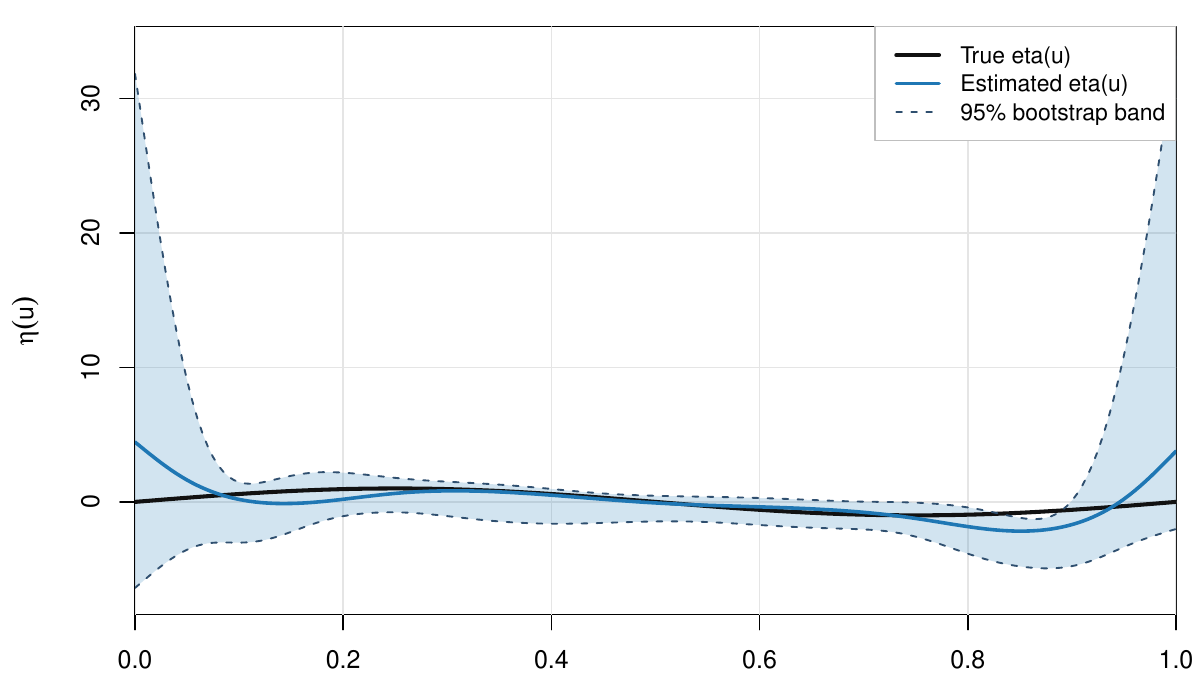}
    \end{minipage}
\caption{Representative fit of $\eta_0(u)$ and $\widehat{\eta}(u)$ with $95\%$ bootstrap bands under Bernoulli case. ($a$) \textit{Top Left}: $n=200$, $\rho=0.0$; ($b$) \textit{Top Right}: $n=200$, $\rho=0.3$; ($c$) \textit{Bottom}: $n=200$, $\rho=0.6$.} \label{fig1} 
\end{figure}

\begin{figure}[!]
\centering
    \begin{minipage}{0.49\linewidth}
    \centering
    \includegraphics[width=1\linewidth]{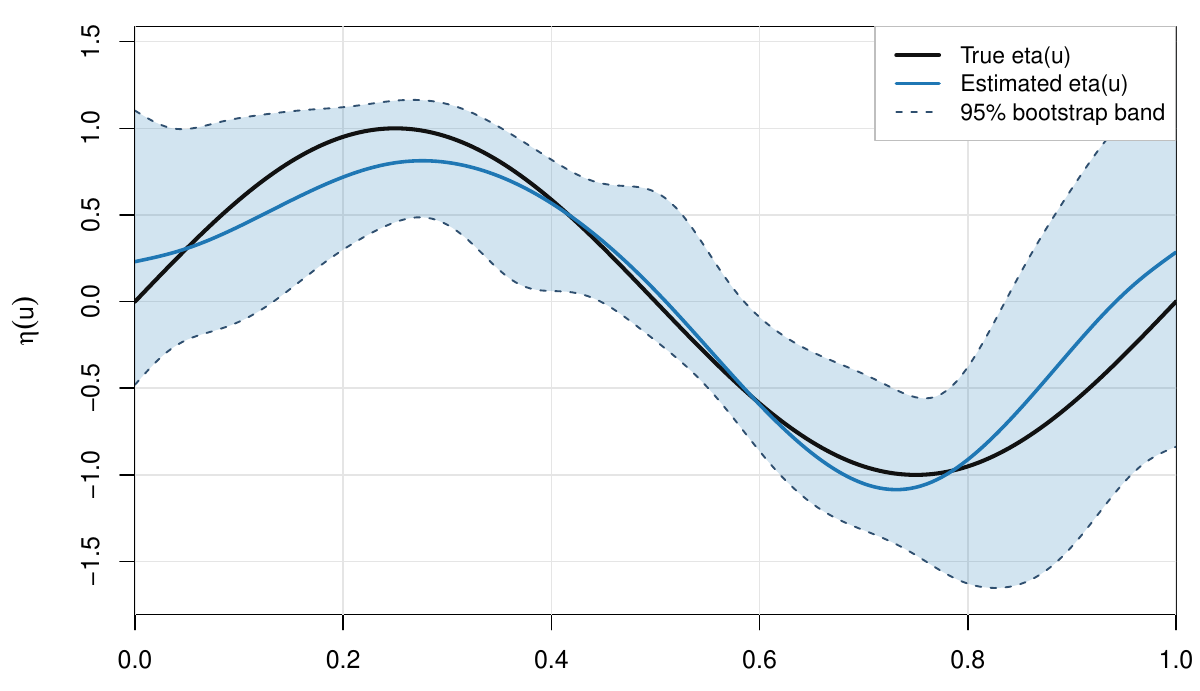}
    \end{minipage}
    \begin{minipage}{0.49\linewidth}
    \centering
    \includegraphics[width=1\linewidth]{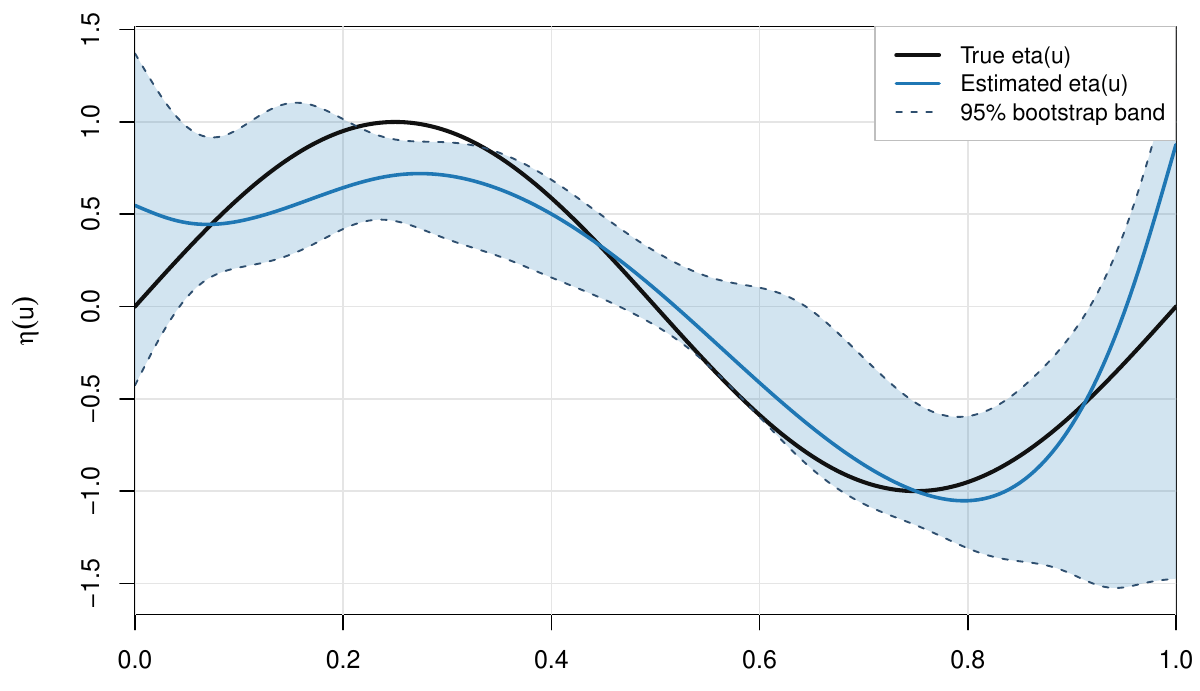}
    \end{minipage}
\caption{Representative fit of $\eta_0(u)$ and $\widehat{\eta}(u)$ with $95\%$ bootstrap bands under Gaussian case. ($a$) \textit{Left}: $n=100$, $\rho=0.0$; ($b$) \textit{Right}: $n=200$, $\rho=0.0$.}\label{fig2}
\end{figure}

\subsection{Sensitivity Analyses}\label{ssec:sim_sens}

We conduct sensitivity analyses to assess how the competing procedures behave when key implementation choices are perturbed. In particular, we vary the working correlation among independence, exchangeable, and AR(1), and we examine robustness across outcome families. Throughout, the spline dimension $K$ is selected by BIC over the same candidate grid used in the main experiments, so the comparison reflects the intended data-driven implementation rather than hand tuning.

Tables~\ref{tab:sim_rmse_gauss}--\ref{tab:sim_rmse_pois} show that the point estimation conclusions are largely insensitive to the working correlation choice, in the sense that the relative ordering of methods in RMSE/ISE-type summaries remains stable across three working correlation structures. In particular, Profile-BEL and GEE-Wald (which share the same estimating-equation backbone) yield comparable point estimates under different working correlations, while Naive-EL is essentially unchanged across correlation settings by construction. The polynomial GEE competitor exhibits the greatest variability, both in magnitude and in rank, especially for the more difficult components and for non-Gaussian responses, indicating that misspecification in its mean/working-structure interplay can translate into noticeably less stable finite-sample performance.

Table~\ref{tab:sim_cover_main} summarizes inference sensitivity via empirical coverage and average CI length. Two patterns emerge. First, Profile-BEL is largely insensitive to the working correlation choice: for each outcome family, its coverages and lengths change only mildly when switching among three working correlation structures, consistent with the blockwise construction based on   estimating equations and calibration that does not hinge on correctly specifying the working correlation. Second, Profile-BEL delivers a more favorable coverage--length trade-off for the harder parameters, most notably for $\beta_1$ and $\alpha_2$. In the Gaussian case, coverage for $\beta_1$ improves relative to Naive-EL and GEE-Wald while keeping CI length moderate, and for $\alpha_2$ Profile-BEL is close to nominal coverage with substantially shorter intervals than GEE-Wald. Similar behavior is seen under Bernoulli outcomes, where Profile-BEL maintains near-nominal coverage for $\alpha_2$ and improves calibration for $\beta_1$ without inflating lengths. Under the Poisson design, inference for $\beta_1$ is challenging for all methods, yet Profile-BEL still dominates the baselines in coverage for $\beta_1$ while remaining well-calibrated for $\beta_2$ and $\alpha_2$. 

Overall, these sensitivity results reinforce the main message: Profile-BEL provides the most robust and practically useful inference across correlation specifications and outcome types, whereas the polynomial GEE alternative can be substantially less stable in calibration and/or interval length for the more difficult components.

\section{Real Data Application}\label{sec:realdata}

\subsection{Data Description}\label{ssec:realdata_data}
We use \texttt{epil} dataset from the \textsf{MASS} package in \textsf{R} for analysis, which is a public longitudinal epilepsy study 1852 observations. The outcome is the seizure count recorded repeatedly for each subject across follow-up periods, which naturally motivates a Poisson-type mean model with within-subject correlation.

Let $Y_{ij}$ denote the seizure count for subject $i$ at visit/period $j$, and $t_{ij}\in[0,1]$ be the rescaled time index. We include a treatment indicator and baseline severity (e.g., baseline seizure frequency) as covariates, together with demographic adjustment variables (e.g., age). All continuous covariates are standardized.

\subsection{Implementation}\label{ssec:realdata_model}
We fit the same semiparametric longitudinal model as in the simulation section:
\begin{equation}\label{eq:realdata_epil}
\log\{\mathbb{E}(Y_{ij}\mid \mathbf{X}_{ij}, t_{ij})\}= \mathbf{X}_{ij}^\top\boldsymbol{\theta} + \eta(t_{ij}),
\end{equation}
where $\boldsymbol{\theta}$ is the finite-dimensional target and $\eta(\cdot)$ is an unknown smooth time effect. We approximate $\eta(t)$ by a spline basis with dimension $K$, and select $K$ by BIC over the same candidate grid used in the simulation (e.g., $K\in\{6,8,10,12\}$). Working correlation is taken as independence, AR(1), and exchangeable, matching the simulation design.

We compare the same four procedures as in Section~\ref{ssec:sim_impl}: \textbf{Profile-BEL}, \textbf{Naive-EL}, \textbf{GEE-Wald} and \textbf{GEE-Poly}. For Profile-BEL, the point estimator updates the nuisance parameter (including $\alpha$) iteratively, and inference for $\pmb{\theta}$ is obtained by profiling the BEL statistic. For the functional component, we construct a $95\%$ bootstrap pointwise band for $\eta(\cdot)$ using   resampling with $B$ bootstrap replicates.

\subsection{Evaluation Criteria}\label{ssec:realdata_metrics}
In the real-data analysis, the true $\boldsymbol{\theta}$ and the unknown link $\eta(\cdot)$ are not observable, so we compare methods using evaluation criteria that avoid reliance on ground truth.  

We therefore assess predictive performance through   $K$-fold cross-validation; for count outcomes, we report the held-out Poisson deviance (equivalently, the negative log-likelihood up to an additive constant), where smaller values indicate better out-of-sample fit.  For scientific interpretation and uncertainty quantification, we focus on key components of $\boldsymbol{\theta}$ (e.g., treatment and baseline severity effects) and summarize each method by its $95\%$ confidence-interval length as well as how sensitive that length is to the working-correlation specification (independence/AR(1)/exchangeable).  Finally, to evaluate uncertainty for the nonparametric component, we compare the estimated shapes of $\eta(\cdot)$ and the widths of the associated $95\%$ pointwise bands: bands that remain overly narrow despite poor predictive performance may indicate underestimated uncertainty, whereas uniformly very wide bands can reflect loss of efficiency.

To explicitly quantify how sensitive interval estimation is to the working-correlation choice, Table~\ref{tab:real_stab} reports a new stability metric, \emph{Range across correlation}. For a given parameter $\theta_j$, define
\[
\text{Range across corr.}
=
\max_{c\in\{\text{ind},\text{ar1},\text{exc}\}}
\Bigl\{\text{CI length}(\theta_j;c)\Bigr\}
-
\min_{c\in\{\text{ind},\text{ar1},\text{exc}\}}
\Bigl\{\text{CI length}(\theta_j;c)\Bigr\}.
\]
Smaller values indicate more stable (i.e., less correlation-sensitive) uncertainty quantification across plausible working correlations.

\subsection{Results}\label{ssec:realdata_results}
We summarize cross-validated deviance, point estimates and CIs for selected coefficients in $\boldsymbol{\theta}$, and the fitted $\eta(t)$ with pointwise bands in Table~\ref{tab:real_cv}-\ref{tab:real_stab}. Figure~\ref{fig:realdata_eta} visualizes $\widehat{\eta}(t)$: Profile-BEL estimate with its $95\%$ bootstrap pointwise band, and overlay the GEE-Poly fit for comparison.

\begin{table}[h]
\centering
\caption{$K$-fold CV Poisson deviance.}
\label{tab:real_cv}
\begin{tabular}{lccc}
\toprule
Method & Independence & AR(1) & Exchangeable \\
\hline
Profile-BEL & 4.5251 & 4.5113 & 5.0705 \\
GEE-Wald    & 4.6351 & 4.6213 & 5.1805 \\
Naive-EL    & 4.6357 & 4.6357 & 4.6357 \\
GEE-Poly    & 4.6129 & 4.5849 & 4.7328 \\
\bottomrule
\end{tabular}
\end{table}

\begin{table}[h]
\centering
\caption{Point estimates and 95\% CIs under different working correlations.}
\label{tab:real_coef}
\begin{tabular}{llccc}
\toprule
Coefficient & Method & Independence & AR(1) & Exchangeable \\
\hline
\multirow{4}{*}{$\theta_{\text{trt}}$}
& Profile-BEL & -0.163 [-0.306, -0.020] & -0.174 [-0.312, -0.037] & -0.023 [-0.219, 0.172] \\
& GEE-Wald    & -0.153 [-0.488, 0.183] & -0.161 [-0.484, 0.161] & -0.005 [-0.464, 0.453] \\
& Naive-EL    & -0.153 [-0.491, 0.186] & -0.153 [-0.491, 0.186] & -0.153 [-0.491, 0.186] \\
& GEE-Poly    & -0.153 [-0.488, 0.183]	 & -0.161 [-0.479, 0.157] & 0.001 [-0.377, 0.379] \\
\hline
\multirow{4}{*}{$\theta_{\text{base}}$}
& Profile-BEL & 0.613 [0.585, 0.640] & 0.627 [0.599, 0.654] & 0.686 [0.641, 0.731] \\
& GEE-Wald    & 0.605 [0.540, 0.670] & 0.617 [0.551, 0.682] & 0.674 [0.568, 0.780] \\
& Naive-EL    & 0.605 [0.540, 0.670] & 0.605 [0.540, 0.670] & 0.605 [0.540, 0.670] \\
& GEE-Poly    & 0.605 [0.540, 0.670] & 0.615 [0.550, 0.680] & 0.648 [0.551, 0.744] \\
\hline
\multirow{4}{*}{$\theta_{\text{age}}$}
& Profile-BEL & 0.148 [0.088, 0.209] & 0.169 [0.107, 0.230] & 0.304 [0.182, 0.426] \\
& GEE-Wald    & 0.142 [0.000, 0.284] & 0.161 [0.015, 0.306] & 0.294 [0.008, 0.580] \\
& Naive-EL    & 0.142 [-0.001, 0.286] & 0.142 [-0.001, 0.286] & 0.142 [-0.001, 0.286] \\
& GEE-Poly    & 0.142 [0.000, 0.284] & 0.158 [0.014, 0.302]	 & 0.248 [-0.015, 0.512] \\
\bottomrule
\end{tabular}
\end{table}

\begin{table}[h]
\centering
\caption{CI length and correlation-sensitivity.}
\label{tab:real_stab}
\begin{tabular}{llcc}
\toprule
Coefficient & Method & Avg. CI length & Range across corr. \\
\hline
\multirow{4}{*}{$\theta_{\text{trt}}$}
& Profile-BEL & 0.3170 & 0.1161 \\
& GEE-Wald    & 0.7441 & 0.2726 \\
& Naive-EL    & 0.6765 & 0.0000 \\
& GEE-Poly    & 0.6878 & 0.1204 \\
\hline
\multirow{4}{*}{$\theta_{\text{base}}$}
& Profile-BEL & 0.0670 & 0.0351 \\
& GEE-Wald    & 0.1572 & 0.0823 \\
& Naive-EL    & 0.1305 & 0.0000 \\
& GEE-Poly    & 0.1507 & 0.0640 \\
\bottomrule
\end{tabular}
\end{table}

\begin{figure}[!]
\centering
\includegraphics[width=0.95\textwidth]{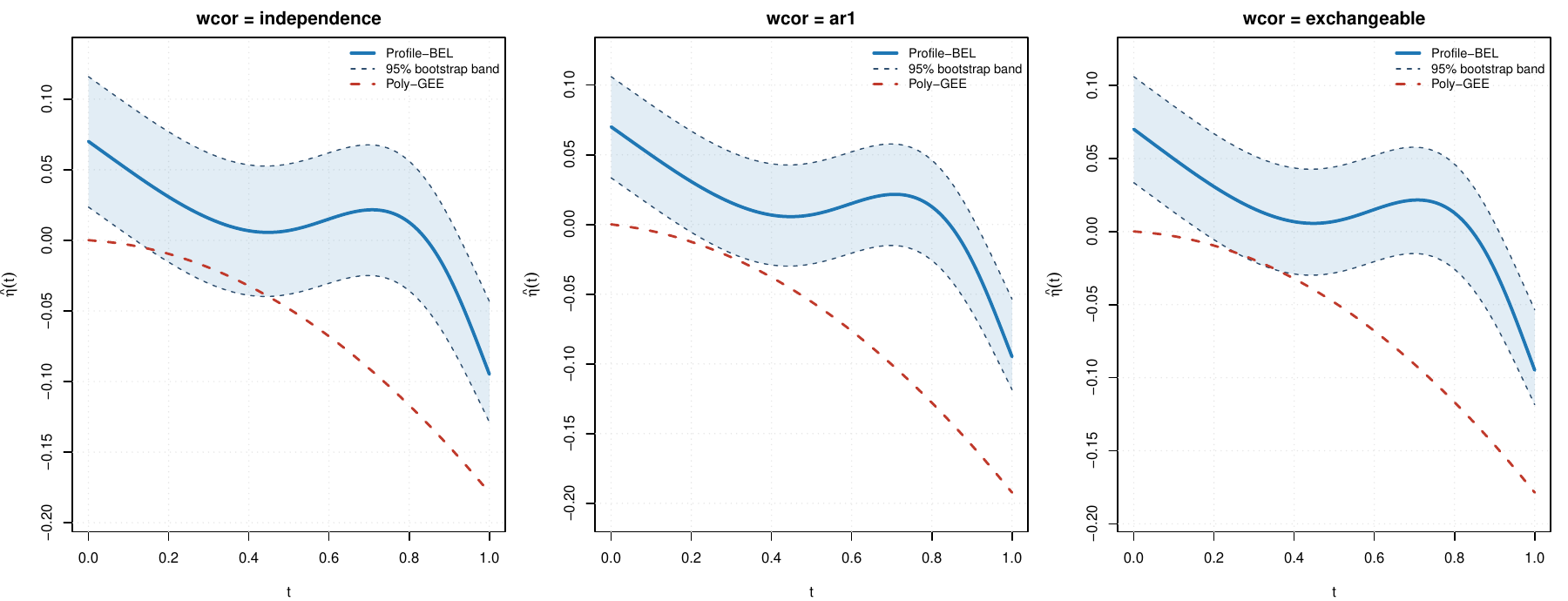}
\caption{Estimated time effect $\widehat{\eta}(t)$ on the epilepsy dataset. All $\eta$ estimates are post-processed to satisfy the same identifiability constraint.}
\label{fig:realdata_eta}
\end{figure}

Across all methods, the treatment-related coefficient $\theta_{\mathrm{trt}}$ is estimated to be negative, but Profile-BEL delivers the most informative inference: under both independence and AR(1), its $95\%$ confidence interval excludes zero while remaining substantially shorter than the corresponding intervals from GEE-Wald and GEE-Poly. In contrast, the Wald-type intervals tend to be wider and often include zero, reflecting their heavier reliance on plug-in variance estimation in a moderate-sample correlated setting, and Naive-EL (by construction) yields essentially identical point/interval outputs regardless of the working correlation because it ignores within-subject dependence. The fact that Profile-BEL simultaneously sharpens uncertainty quantification and preserves the direction and magnitude of the estimated effect provides empirical support for the self-normalization advantage of empirical likelihood in this longitudinal semiparametric problem.

Beyond significance, the real-data analysis highlights stability and practical robustness. The cross-validated deviance favors Profile-BEL under plausible working correlations, indicating that properly accounting for dependence while profiling out the nonparametric link improves generalization rather than merely tightening intervals. Moreover, the CI-stability summaries show that Profile-BEL achieves the best length--stability trade-off: its average CI length is the smallest among competitors, and its range across correlation specifications is comparatively limited, suggesting reduced sensitivity to the working correlation choice. 

Taken together, the real-data findings align with the simulation evidence and reinforce the main takeaway of the paper: Profile-BEL offers a practical, dependence-aware inference strategy that produces tighter and more stable conclusions for longitudinal GPLSIMs.
\section{Conclusion}\label{sec:conclusion}

We studied a generalized partially linear single-index model for longitudinal data, where the covariate effects are decomposed into a finite-dimensional linear component and an unknown smooth link along a low-dimensional index. By combining spline-sieve profiling with   estimating equations, we developed a practical inference framework that accommodates within-subject dependence while retaining a clear separation between the target parameter $\pmb{\theta}$ and the nuisance function $\eta(\cdot)$. The proposed profile block empirical likelihood provides likelihood-type confidence regions for $\pmb{\theta}$ without requiring explicit stabilization of sandwich variance estimators, and our asymptotic theory establishes a Wilks-type chi-square limit under mild regularity conditions.

From a methodological perspective, the main advantage of the proposed approach is its flexibility for longitudinal dependence. The inference is constructed at the subject level, which naturally respects the block structure of repeated measurements, and it remains applicable when the working correlation is only an approximation of the true dependence. The spline-based profiling step offers a convenient and computationally efficient way to estimate the unknown link function, while preserving root-$n$ inference for the finite-dimensional component. Our simulation results support these theoretical findings and suggest that the empirical-likelihood calibration can deliver stable coverage in moderate samples across a range of outcome types and correlation strengths.

Several extensions are of interest. First, the current framework assumes a bounded cluster size; it would be useful to study regimes where the number of repeated measurements grows with $n$, potentially requiring refined empirical-process arguments and alternative normalization. Second, one may incorporate more flexible dependence models, including time-varying correlation or random-effect structures, while retaining the   estimating-equation foundation. Third, the index structure can be enriched by allowing multiple indices, leading to an additive multi-index link $\sum_{\ell=1}^{L}\eta_{\ell}(\pmb{z}^{\top}\pmb{\alpha}_{\ell})$ that balances interpretability and flexibility. Fourth, it is natural to consider high-dimensional linear components with structured regularization, where one can combine profiling with penalized estimating equations to enable variable selection in the presence of an unknown link function. Finally, extending the current methodology to handle irregular observation times, missingness mechanisms, and more complex measurement error structures would broaden its applicability in real longitudinal studies.

Overall, the proposed profile block empirical likelihood framework offers a principled and implementable route for inference in semiparametric longitudinal models with dimension reduction. We hope it will serve as a useful building block for more general dependence structures and richer functional components in future work.

\bibliographystyle{erae}
\bibliography{refs}

\clearpage
\appendix

\section{Proofs}\label{app:proofs}

The proofs of Section~\ref{sec:theory} rely on spline sieve approximation and uniform convergence of the profiled nuisance estimator, a linearization of the profile estimating equation, and a standard EL Lagrange-multiplier expansion.  Throughout the appendix, $C$ denotes a generic positive constant that may change from line to line. All stochastic orders are with respect to $n\to\infty$.

\subsection{Auxiliary Lemmas}

\begin{lemma}[Spline approximation]\label{lem:spline_approx}
Under Assumptions~\ref{ass:smooth}--\ref{ass:index_density}, there exists a coefficient vector $\pmb{\gamma}_{0}^{\ast}=\pmb{\gamma}_{0}^{\ast}(K)$ such that
\[
\sup_{u\in\mathcal{U}}\left|\eta_0(u)-\pmb{B}(u)^{\top}\pmb{\gamma}_{0}^{\ast}\right|
\le C K^{-s}.
\]
\end{lemma}

\begin{proof}
This is a standard property of polynomial spline approximation on a compact interval with quasi-uniform knots.
Since $\eta_0$ is $s$ times continuously differentiable with bounded $s$-th derivative (Assumption~\ref{ass:smooth}),
the spline space of dimension $K$ contains an approximant with sup-norm error $O(K^{-s})$.
A detailed construction can be found in \citep{huang2004polynomial,he2000parameters}.
\end{proof}

\begin{lemma}[Uniform consistency of the profiled nuisance]\label{lem:gamma_uniform}
Under Assumptions~\ref{ass:sampling}--\ref{ass:V}. Let $\widehat{\pmb{\gamma}}(\pmb{\theta})$ solve \eqref{eq:gamma_prof} and $\pmb{\gamma}_0(\pmb{\theta})$ solve \eqref{eq:gamma_pop}. Then, uniformly over $\pmb{\theta}$ in a neighborhood $\mathcal{N}$ of $\pmb{\theta}_0$,
\[
\left\|\widehat{\pmb{\gamma}}(\pmb{\theta})-\pmb{\gamma}_0(\pmb{\theta})\right\|_2
=
O_{\mathbb{P}}\!\left(\sqrt{\frac{K}{n}}+K^{-s}\right).
\]
Consequently, uniformly over $\pmb{\theta}\in\mathcal{N}$,
\[
\sup_{u\in\mathcal{U}}
\left|
\pmb{B}(u)^{\top}\widehat{\pmb{\gamma}}(\pmb{\theta})
-
\pmb{B}(u)^{\top}\pmb{\gamma}_0(\pmb{\theta})
\right|
=
O_{\mathbb{P}}\!\left(\sqrt{\frac{K}{n}}+K^{-s}\right).
\]
\end{lemma}

\begin{proof}
Fix a neighborhood $\mathcal{N}$ of $\pmb{\theta}_0$, $\forall\pmb{\theta}\in\mathcal{N}$, define the sample and population inner estimating maps
\[
\pmb{\Phi}_n(\pmb{\gamma};\pmb{\theta})
=
\frac{1}{n}\sum_{i=1}^{n}
\pmb{B}_i(\pmb{\theta})^{\top}
\pmb{\Delta}_i(\pmb{\theta},\pmb{\gamma})
\pmb{V}_i(\pmb{\theta},\pmb{\gamma})^{-1}
\left\{\pmb{Y}_i-\pmb{\mu}_i(\pmb{\theta},\pmb{\gamma})\right\},
\]
\[
\pmb{\Phi}_0(\pmb{\gamma};\pmb{\theta})
=
\mathbb{E}\!\left[
\pmb{B}_i(\pmb{\theta})^{\top}
\pmb{\Delta}_i(\pmb{\theta},\pmb{\gamma})
\pmb{V}_i(\pmb{\theta},\pmb{\gamma})^{-1}
\left\{\pmb{Y}_i-\pmb{\mu}_i(\pmb{\theta},\pmb{\gamma})\right\}
\right].
\]
By definition, $\widehat{\pmb{\gamma}}(\pmb{\theta})$ satisfies $\pmb{\Phi}_n(\widehat{\pmb{\gamma}}(\pmb{\theta});\pmb{\theta})=\pmb{0}$ and $\pmb{\gamma}_0(\pmb{\theta})$ satisfies $\pmb{\Phi}_0(\pmb{\gamma}_0(\pmb{\theta});\pmb{\theta})=\pmb{0}$.

We first control the uniform stochastic fluctuation of $\pmb{\Phi}_n$ around $\pmb{\Phi}_0$. Under Assumptions~\ref{ass:sampling}, \ref{ass:moments}, \ref{ass:link}, and \ref{ass:V}, together with $\max_i m_i\le M$, the   contributions are i.i.d.\ across $i$ with an integrable envelope, and the maps $(\pmb{\theta},\pmb{\gamma})\mapsto \pmb{\Phi}_n(\pmb{\gamma};\pmb{\theta})$ are uniformly Lipschitz in $(\pmb{\theta},\pmb{\gamma})$ on $\mathcal{N}\times \Gamma_K$ for any bounded set $\Gamma_K\subset\mathbb{R}^K$ containing the relevant solutions. Moreover, because the effective dimension of the sieve component is $K$, standard empirical-process bounds for finite-dimensional sieve scores yield
\begin{equation}\label{eq:Phi_uln}
\sup_{\pmb{\theta}\in\mathcal{N}}\sup_{\pmb{\gamma}\in\Gamma_K}
\left\|
\pmb{\Phi}_n(\pmb{\gamma};\pmb{\theta})-\pmb{\Phi}_0(\pmb{\gamma};\pmb{\theta})
\right\|_2
=
O_{\mathbb{P}}\!\left(\sqrt{\frac{K}{n}}\right).
\end{equation}

Next we establish local invertibility (uniformly in $\pmb{\theta}$) of the population Jacobian w.r.t $\pmb{\gamma}$. Denote
\[
\pmb{J}_0(\pmb{\gamma};\pmb{\theta})
=
\frac{\partial}{\partial \pmb{\gamma}^{\top}}\pmb{\Phi}_0(\pmb{\gamma};\pmb{\theta}).
\]
By Assumptions~\ref{ass:link} and \ref{ass:V}, $\pmb{J}_0(\pmb{\gamma};\pmb{\theta})$ exists and is continuous in $(\pmb{\gamma},\pmb{\theta})$. Furthermore, since $\pmb{\Phi}_0(\pmb{\gamma};\pmb{\theta})$ is a generalized least-squares normal equation in $\pmb{\gamma}$, the matrix $-\pmb{J}_0(\pmb{\gamma}_0(\pmb{\theta});\pmb{\theta})$ is positive definite and its eigenvalues are uniformly bounded away from $0$ and $\infty$ over $\pmb{\theta}\in\mathcal{N}$. Consequently, there exists $c>0$ such that
\begin{equation}\label{eq:J0_inv}
\sup_{\pmb{\theta}\in\mathcal{N}}
\left\|
\pmb{J}_0(\pmb{\gamma}_0(\pmb{\theta});\pmb{\theta})^{-1}
\right\|
\le c.
\end{equation}
By the same continuity argument and \eqref{eq:Phi_uln}, the sample Jacobian $\pmb{J}_n(\pmb{\gamma};\pmb{\theta})=\partial \pmb{\Phi}_n(\pmb{\gamma};\pmb{\theta})/\partial \pmb{\gamma}^{\top}$ converges uniformly to $\pmb{J}_0(\pmb{\gamma};\pmb{\theta})$ on $\mathcal{N}\times\Gamma_K$, hence is invertible uniformly with probability tending to one.

We now relate $\widehat{\pmb{\gamma}}(\pmb{\theta})$ to $\pmb{\gamma}_0(\pmb{\theta})$. For any fixed $\pmb{\theta}\in\mathcal{N}$, apply the mean value theorem to $\pmb{\Phi}_n(\widehat{\pmb{\gamma}}(\pmb{\theta});\pmb{\theta})-\pmb{\Phi}_n(\pmb{\gamma}_0(\pmb{\theta});\pmb{\theta})$: there exists $\tilde{\pmb{\gamma}}(\pmb{\theta})$ on the segment joining $\widehat{\pmb{\gamma}}(\pmb{\theta})$ and $\pmb{\gamma}_0(\pmb{\theta})$ such that
\[
\pmb{0}
=
\pmb{\Phi}_n(\widehat{\pmb{\gamma}}(\pmb{\theta});\pmb{\theta})
=
\pmb{\Phi}_n(\pmb{\gamma}_0(\pmb{\theta});\pmb{\theta})
+
\pmb{J}_n(\tilde{\pmb{\gamma}}(\pmb{\theta});\pmb{\theta})
\left\{\widehat{\pmb{\gamma}}(\pmb{\theta})-\pmb{\gamma}_0(\pmb{\theta})\right\}.
\]
Therefore,
\[
\widehat{\pmb{\gamma}}(\pmb{\theta})-\pmb{\gamma}_0(\pmb{\theta})
=
-\pmb{J}_n(\tilde{\pmb{\gamma}}(\pmb{\theta});\pmb{\theta})^{-1}
\pmb{\Phi}_n(\pmb{\gamma}_0(\pmb{\theta});\pmb{\theta}).
\]
Decompose the right-hand side as
\[
\pmb{\Phi}_n(\pmb{\gamma}_0(\pmb{\theta});\pmb{\theta})
=
\left\{\pmb{\Phi}_n(\pmb{\gamma}_0(\pmb{\theta});\pmb{\theta})-\pmb{\Phi}_0(\pmb{\gamma}_0(\pmb{\theta});\pmb{\theta})\right\}
+
\pmb{\Phi}_0(\pmb{\gamma}_0(\pmb{\theta});\pmb{\theta}).
\]
The second term is exactly $\pmb{0}$ by definition of $\pmb{\gamma}_0(\pmb{\theta})$, hence
\[
\pmb{\Phi}_n(\pmb{\gamma}_0(\pmb{\theta});\pmb{\theta})
=
\pmb{\Phi}_n(\pmb{\gamma}_0(\pmb{\theta});\pmb{\theta})-\pmb{\Phi}_0(\pmb{\gamma}_0(\pmb{\theta});\pmb{\theta}).
\]
Combining this with the uniform invertibility of $\pmb{J}_n$ and the uniform bound \eqref{eq:Phi_uln} gives, uniformly for $\pmb{\theta}\in\mathcal{N}$,
\[
\left\|\widehat{\pmb{\gamma}}(\pmb{\theta})-\pmb{\gamma}_0(\pmb{\theta})\right\|_2
=
O_{\mathbb{P}}\!\left(\sqrt{\frac{K}{n}}\right).
\]

To incorporate the sieve approximation error, let $\pmb{\gamma}_0^{\ast}(K)$ be the spline coefficient vector in Lemma~\ref{lem:spline_approx} so that $\sup_{u\in\mathcal{U}}|\eta_0(u)-\pmb{B}(u)^{\top}\pmb{\gamma}_0^{\ast}(K)|\le C K^{-s}$. Under Assumptions~\ref{ass:link} and \ref{ass:V}, the map $\pmb{\gamma}\mapsto \pmb{\Phi}_0(\pmb{\gamma};\pmb{\theta})$ is continuously differentiable with Jacobian uniformly nonsingular around $\pmb{\gamma}_0(\pmb{\theta})$. 

Because the population equation \eqref{eq:gamma_pop} is the sieve-restricted population score, the deviation between its root $\pmb{\gamma}_0(\pmb{\theta})$ and the oracle approximation $\pmb{\gamma}_0^{\ast}(K)$ is controlled by the approximation error through a standard implicit-function perturbation argument, yielding an additional $O(K^{-s})$ term. Hence, uniformly on $\mathcal{N}$,
\[
\left\|\widehat{\pmb{\gamma}}(\pmb{\theta})-\pmb{\gamma}_0(\pmb{\theta})\right\|_2
=
O_{\mathbb{P}}\!\left(\sqrt{\frac{K}{n}}+K^{-s}\right).
\]

Finally, the stated sup-norm bound follows from boundedness of the spline basis on $\mathcal{U}$:
\[
\sup_{u\in\mathcal{U}}
\left|
\pmb{B}(u)^{\top}\widehat{\pmb{\gamma}}(\pmb{\theta})
-
\pmb{B}(u)^{\top}\pmb{\gamma}_0(\pmb{\theta})
\right|
\le
\sup_{u\in\mathcal{U}}\|\pmb{B}(u)\|_2\cdot
\left\|\widehat{\pmb{\gamma}}(\pmb{\theta})-\pmb{\gamma}_0(\pmb{\theta})\right\|_2,
\]
which is
$O_{\mathbb{P}}(\sqrt{K/n}+K^{-s})$ uniformly over $\pmb{\theta}\in\mathcal{N}$.
\end{proof}

\begin{lemma}[Linearization of the profile estimating equation]\label{lem:U_linear}
Under Assumptions~\ref{ass:sampling}--\ref{ass:H}, the sample profile map $\pmb{U}_n(\pmb{\theta})=n^{-1}\sum_{i=1}^{n}\pmb{g}_i(\pmb{\theta})$ satisfies, uniformly for $\pmb{\theta}$ in a neighborhood of $\pmb{\theta}_0$,
\[
\pmb{U}_n(\pmb{\theta})
=
\pmb{U}_n(\pmb{\theta}_0)
+\pmb{H}_0(\pmb{\theta}-\pmb{\theta}_0)
+o_{\mathbb{P}}\!\left(\|\pmb{\theta}-\pmb{\theta}_0\|_2\right)
+o_{\mathbb{P}}(n^{-1/2}).
\]
Moreover, $\sqrt{n}\,\pmb{U}_n(\pmb{\theta}_0)\Rightarrow \mathcal{N}(\pmb{0},\pmb{S}_0)$.
\end{lemma}

\begin{proof}
Write $\pmb{U}_n(\pmb{\theta})=n^{-1}\sum_{i=1}^{n}\pmb{g}_i(\pmb{\theta})$ and $\pmb{U}_{n,0}(\pmb{\theta})=n^{-1}\sum_{i=1}^{n}\pmb{g}_{i,0}(\pmb{\theta})$. We first show that replacing the profiled nuisance estimator $\widehat{\pmb{\gamma}}(\pmb{\theta})$ by its population counterpart $\pmb{\gamma}_0(\pmb{\theta})$ only induces an asymptotically negligible error at the $n^{-1/2}$ scale. 

By Lemma~\ref{lem:gamma_uniform} and the smoothness of the map $(\pmb{\theta},\pmb{\gamma})\mapsto \pmb{g}_i(\pmb{\theta})$ implied by Assumptions~\ref{ass:link} and \ref{ass:V}, there exists a neighborhood $\mathcal{N}$ of $\pmb{\theta}_0$ such that
\[
\sup_{\pmb{\theta}\in\mathcal{N}}
\left\|
\pmb{U}_n(\pmb{\theta})-\pmb{U}_{n,0}(\pmb{\theta})
\right\|_2
=
O_{\mathbb{P}}\!\left(\sqrt{\frac{K}{n}}+K^{-s}\right).
\]
Under Assumption~\ref{ass:K}, the right-hand side is $o_{\mathbb{P}}(n^{-1/2})$, hence uniformly on $\mathcal{N}$,
\begin{equation}\label{eq:Un_replace}
\pmb{U}_n(\pmb{\theta})
=
\pmb{U}_{n,0}(\pmb{\theta})+o_{\mathbb{P}}(n^{-1/2}).
\end{equation}

Next we linearize $\pmb{U}_{n,0}(\pmb{\theta})$ around $\pmb{\theta}_0$.
Let $\pmb{H}_{i,0}(\pmb{\theta})=\partial \pmb{g}_{i,0}(\pmb{\theta})/\partial \pmb{\theta}^{\top}$.
By Assumptions~\ref{ass:moments}--\ref{ass:V}, $\pmb{H}_{i,0}(\pmb{\theta})$ exists, is continuous in $\pmb{\theta}$, and is uniformly bounded on $\mathcal{N}$.
Therefore, for any $\pmb{\theta}\in\mathcal{N}$, a mean-value expansion yields
\[
\pmb{U}_{n,0}(\pmb{\theta})
=
\pmb{U}_{n,0}(\pmb{\theta}_0)
+
\left\{\frac{1}{n}\sum_{i=1}^{n}\pmb{H}_{i,0}(\tilde{\pmb{\theta}})\right\}
(\pmb{\theta}-\pmb{\theta}_0),
\]
where $\tilde{\pmb{\theta}}$ lies on the line segment joining $\pmb{\theta}$ and $\pmb{\theta}_0$.
By a uniform law of large numbers (using Assumptions~\ref{ass:sampling} and \ref{ass:moments}),
\[
\sup_{\pmb{\theta}\in\mathcal{N}}
\left\|
\frac{1}{n}\sum_{i=1}^{n}\pmb{H}_{i,0}(\pmb{\theta})-\pmb{H}_0
\right\|
=o_{\mathbb{P}}(1),
\]
so the derivative term can be replaced by $\pmb{H}_0$ up to $o_{\mathbb{P}}(1)$ uniformly on $\mathcal{N}$.
Consequently,
\[
\pmb{U}_{n,0}(\pmb{\theta})
=
\pmb{U}_{n,0}(\pmb{\theta}_0)
+\pmb{H}_0(\pmb{\theta}-\pmb{\theta}_0)
+o_{\mathbb{P}}\!\left(\|\pmb{\theta}-\pmb{\theta}_0\|_2\right),
\]
uniformly for $\pmb{\theta}\in\mathcal{N}$.

Finally we establish the central limit theorem at $\pmb{\theta}_0$.
Since subjects are independent (Assumption~\ref{ass:sampling}) and $\mathbb{E}\|\pmb{g}_{i,0}(\pmb{\theta}_0)\|_2^2<\infty$ (Assumption~\ref{ass:moments}),
the multivariate Lindeberg--Feller CLT yields
\[
\sqrt{n}\,\pmb{U}_{n,0}(\pmb{\theta}_0)
=
\frac{1}{\sqrt{n}}\sum_{i=1}^{n}\pmb{g}_{i,0}(\pmb{\theta}_0)
\Rightarrow
\mathcal{N}(\pmb{0},\pmb{S}_0).
\]
Combining this CLT with \eqref{eq:Un_replace} shows
$\sqrt{n}\,\pmb{U}_n(\pmb{\theta}_0)\Rightarrow \mathcal{N}(\pmb{0},\pmb{S}_0)$, and substituting \eqref{eq:Un_replace} into the above linearization of $\pmb{U}_{n,0}(\pmb{\theta})$ yields
\[
\pmb{U}_n(\pmb{\theta})
=
\pmb{U}_n(\pmb{\theta}_0)+\pmb{H}_0(\pmb{\theta}-\pmb{\theta}_0)
+o_{\mathbb{P}}\!\left(\|\pmb{\theta}-\pmb{\theta}_0\|_2\right)
+o_{\mathbb{P}}(n^{-1/2}),
\]
uniformly on $\mathcal{N}$. This completes the proof.
\end{proof}

\subsection{Proofs of Theorems}\label{app:proofs_main}

In this subsection we provide proofs for Theorems~\ref{thm:consistency}--\ref{thm:wilks}. Throughout, $\mathcal{N}$ denotes a sufficiently small open neighborhood of $\pmb{\theta}_0$ on which the expansions in Lemmas~\ref{lem:gamma_uniform} and \ref{lem:U_linear} hold uniformly, and where Assumption~\ref{ass:H} guarantees local uniqueness of the population root.

\begin{proof}[Proof of Theorem~\ref{thm:consistency}]
We prove existence of a root $\widehat{\pmb{\theta}}\in\mathcal{N}$ to \eqref{eq:Utheta} with probability tending to one, $\widehat{\pmb{\theta}}\to \pmb{\theta}_0$, and the stated rates for $\widehat{\pmb{\theta}}$ and $\widehat{\eta}$.

\paragraph{Existence and consistency of a local root.}
Let $\pmb{U}_n(\pmb{\theta})=n^{-1}\sum_{i=1}^{n}\pmb{g}_i(\pmb{\theta})$ be the sample profile estimating map so that \eqref{eq:Utheta} is $\pmb{U}_n(\pmb{\theta})=\pmb{0}$, and let $\pmb{U}_0(\pmb{\theta})=\mathbb{E}\{\pmb{g}_{i,0}(\pmb{\theta})\}$ be its population counterpart. 

By Assumption~\ref{ass:H}, $\pmb{U}_0(\pmb{\theta})$ is continuously differentiable on $\mathcal{N}$, has a unique zero at $\pmb{\theta}_0$, and $\pmb{H}_0=\partial \pmb{U}_0(\pmb{\theta})/\partial \pmb{\theta}^{\top}|_{\pmb{\theta}=\pmb{\theta}_0}$ is nonsingular. Thus, there exists $\varepsilon>0$ such that on the sphere $\partial\mathcal{B}(\pmb{\theta}_0,\varepsilon)$,
\begin{equation}\label{eq:U0_separation}
\inf_{\|\pmb{\theta}-\pmb{\theta}_0\|_2=\varepsilon}\|\pmb{U}_0(\pmb{\theta})\|_2
\ge c_0>0.
\end{equation}
Lemma~\ref{lem:U_linear} implies $\sup_{\pmb{\theta}\in\mathcal{N}}\|\pmb{U}_n(\pmb{\theta})-\pmb{U}_0(\pmb{\theta})\|_2=o_{\mathbb{P}}(1)$. Hence, with probability tending to one,
\[
\sup_{\|\pmb{\theta}-\pmb{\theta}_0\|_2=\varepsilon}\|\pmb{U}_n(\pmb{\theta})-\pmb{U}_0(\pmb{\theta})\|_2 \le c_0/2,
\]
which together with \eqref{eq:U0_separation} gives
\[
\inf_{\|\pmb{\theta}-\pmb{\theta}_0\|_2=\varepsilon}\|\pmb{U}_n(\pmb{\theta})\|_2 \ge c_0/2.
\]
Since $\pmb{U}_n(\pmb{\theta})$ is continuous in $\pmb{\theta}$ on $\mathcal{N}$ (by Assumptions~\ref{ass:link}--\ref{ass:V} and the profile construction), a standard topological argument for vector equations (e.g., degree theory for Z-estimators) implies that $\pmb{U}_n(\pmb{\theta})=\pmb{0}$ admits at least one root $\widehat{\pmb{\theta}}$ inside $\mathcal{B}(\pmb{\theta}_0,\varepsilon)$. Therefore $\widehat{\pmb{\theta}}\to \pmb{\theta}_0$ in probability.

\paragraph{Root-$n$ rate for $\widehat{\pmb{\theta}}$.}
Because $\widehat{\pmb{\theta}}\in\mathcal{N}$ w.h.p., we may apply Lemma~\ref{lem:U_linear} at $\pmb{\theta}=\widehat{\pmb{\theta}}$:
\begin{equation}\label{eq:Un_expand_at_that}
\pmb{0}
=
\pmb{U}_n(\widehat{\pmb{\theta}})
=
\pmb{U}_n(\pmb{\theta}_0)+\pmb{H}_0(\widehat{\pmb{\theta}}-\pmb{\theta}_0)
+\pmb{r}_n,
\end{equation}
where $\|\pmb{r}_n\|_2=o_{\mathbb{P}}(\|\widehat{\pmb{\theta}}-\pmb{\theta}_0\|_2)+o_{\mathbb{P}}(n^{-1/2})$. Lemma~\ref{lem:U_linear} also yields $\pmb{U}_n(\pmb{\theta}_0)=O_{\mathbb{P}}(n^{-1/2})$. Left-multiplying \eqref{eq:Un_expand_at_that} by $\pmb{H}_0^{-1}$ and taking norms gives
\[
\|\widehat{\pmb{\theta}}-\pmb{\theta}_0\|_2
\le
\|\pmb{H}_0^{-1}\|\cdot\|\pmb{U}_n(\pmb{\theta}_0)\|_2
+
\|\pmb{H}_0^{-1}\|\cdot\|\pmb{r}_n\|_2.
\]
Since the first term is $O_{\mathbb{P}}(n^{-1/2})$ and the remainder is asymptotically smaller than $\|\widehat{\pmb{\theta}}-\pmb{\theta}_0\|_2$ plus $n^{-1/2}$, a standard contraction argument implies $\|\widehat{\pmb{\theta}}-\pmb{\theta}_0\|_2=O_{\mathbb{P}}(n^{-1/2})$.

\paragraph{Uniform rate for $\widehat{\eta}$.}
Recall $\widehat{\eta}(u)=\pmb{B}(u)^{\top}\widehat{\pmb{\gamma}}(\widehat{\pmb{\theta}})$. Add and subtract $\pmb{B}(u)^{\top}\pmb{\gamma}_0(\widehat{\pmb{\theta}})$ and $\pmb{B}(u)^{\top}\pmb{\gamma}_0(\pmb{\theta}_0)$:
\begin{align*}
\sup_{u\in\mathcal{U}}|\widehat{\eta}(u)-\eta_0(u)|
&\le
\sup_{u}|\pmb{B}(u)^{\top}\{\widehat{\pmb{\gamma}}(\widehat{\pmb{\theta}})-\pmb{\gamma}_0(\widehat{\pmb{\theta}})\}|
+
\sup_{u}|\pmb{B}(u)^{\top}\{\pmb{\gamma}_0(\widehat{\pmb{\theta}})-\pmb{\gamma}_0(\pmb{\theta}_0)\}| \\
&+
\sup_{u}|\pmb{B}(u)^{\top}\pmb{\gamma}_0(\pmb{\theta}_0)-\eta_0(u)|.
\end{align*}
The first term is $O_{\mathbb{P}}(\sqrt{K/n}+K^{-s})$ uniformly by Lemma~\ref{lem:gamma_uniform} and boundedness of $\pmb{B}(u)$ on $\mathcal{U}$. 

For the second term, the map $\pmb{\theta}\mapsto \pmb{\gamma}_0(\pmb{\theta})$ is locally Lipschitz by the implicit function theorem, since $\pmb{\Phi}_0(\pmb{\gamma};\pmb{\theta})=\pmb{0}$ has a locally unique solution and its Jacobian w.r.t.\ $\pmb{\gamma}$ is uniformly nonsingular (see the Jacobian argument in the proof of Lemma~\ref{lem:gamma_uniform}). Thus $\|\pmb{\gamma}_0(\widehat{\pmb{\theta}})-\pmb{\gamma}_0(\pmb{\theta}_0)\|_2\le C\|\widehat{\pmb{\theta}}-\pmb{\theta}_0\|_2=O_{\mathbb{P}}(n^{-1/2})$, so the second term is $O_{\mathbb{P}}(n^{-1/2})$. 

The last term is the sieve approximation error and is $O(K^{-s})$ by Lemma~\ref{lem:spline_approx}. Combining these bounds yields
\[
\sup_{u\in\mathcal{U}}|\widehat{\eta}(u)-\eta_0(u)|
=
O_{\mathbb{P}}\!\left(K^{-s}+\sqrt{\frac{K}{n}}\right),
\]
which completes the proof.
\end{proof}

\begin{proof}[Proof of Theorem~\ref{thm:normal}] From \eqref{eq:Un_expand_at_that} in the previous proof,
\[
\sqrt{n}(\widehat{\pmb{\theta}}-\pmb{\theta}_0)
=
-\pmb{H}_0^{-1}\sqrt{n}\,\pmb{U}_n(\pmb{\theta}_0)
-\pmb{H}_0^{-1}\sqrt{n}\,\pmb{r}_n.
\]
It suffices to show $\sqrt{n}\,\pmb{r}_n=o_{\mathbb{P}}(1)$. By Lemma~\ref{lem:U_linear}, $\|\pmb{r}_n\|_2=o_{\mathbb{P}}(\|\widehat{\pmb{\theta}}-\pmb{\theta}_0\|_2)+o_{\mathbb{P}}(n^{-1/2})$. Using Theorem~\ref{thm:consistency}, $\|\widehat{\pmb{\theta}}-\pmb{\theta}_0\|_2=O_{\mathbb{P}}(n^{-1/2})$, hence $\|\pmb{r}_n\|_2=o_{\mathbb{P}}(n^{-1/2})$ and therefore $\sqrt{n}\,\pmb{r}_n=o_{\mathbb{P}}(1)$.

Lemma~\ref{lem:U_linear} further gives $\sqrt{n}\,\pmb{U}_n(\pmb{\theta}_0)\Rightarrow \mathcal{N}(\pmb{0},\pmb{S}_0)$. By Slutsky's theorem,
\[
\sqrt{n}(\widehat{\pmb{\theta}}-\pmb{\theta}_0)
\Rightarrow
\mathcal{N}\!\left(\pmb{0},\,\pmb{H}_0^{-1}\pmb{S}_0(\pmb{H}_0^{-1})^{\top}\right),
\]
and the influence representation follows immediately by writing $\sqrt{n}\,\pmb{U}_n(\pmb{\theta}_0)=n^{-1/2}\sum_{i=1}^{n}\pmb{g}_{i,0}(\pmb{\theta}_0)+o_{\mathbb{P}}(1)$, which is part of Lemma~\ref{lem:U_linear}.
\end{proof}

\begin{lemma}[Existence and expansion of the EL multiplier]\label{lem:lambda_expand}
Under Assumptions~\ref{ass:sampling}--\ref{ass:H}, with probability tending to one, the Lagrange multiplier equation \eqref{eq:lambda_eq} at $\pmb{\theta}=\pmb{\theta}_0$ admits a unique solution $\widehat{\pmb{\lambda}}\in\mathbb{R}^{d}$ such that $\|\widehat{\pmb{\lambda}}\|_2=O_{\mathbb{P}}(n^{-1/2})$, and
\[
\widehat{\pmb{\lambda}}
=
\pmb{S}_n(\pmb{\theta}_0)^{-1}\bar{\pmb{g}}(\pmb{\theta}_0)
+o_{\mathbb{P}}(n^{-1/2}),
\]
where $\bar{\pmb{g}}(\pmb{\theta}_0)=n^{-1}\sum_{i=1}^{n}\pmb{g}_i(\pmb{\theta}_0)$ and $\pmb{S}_n(\pmb{\theta}_0)=n^{-1}\sum_{i=1}^{n}\pmb{g}_i(\pmb{\theta}_0)\pmb{g}_i(\pmb{\theta}_0)^{\top}$.
\end{lemma}

\begin{proof}
For brevity write $\pmb{g}_i=\pmb{g}_i(\pmb{\theta}_0)$, $\bar{\pmb{g}}=n^{-1}\sum_{i=1}^{n}\pmb{g}_i$, and $\pmb{S}_n=n^{-1}\sum_{i=1}^{n}\pmb{g}_i\pmb{g}_i^{\top}$. Define the dual map
\[
\pmb{\Psi}(\pmb{\lambda})
=
\frac{1}{n}\sum_{i=1}^{n}\frac{\pmb{g}_i}{1+\pmb{\lambda}^{\top}\pmb{g}_i},
\]
so that \eqref{eq:lambda_eq} is equivalent to $\pmb{\Psi}(\pmb{\lambda})=\pmb{0}$.

We first record the size of the key empirical quantities. By Lemma~\ref{lem:U_linear}, $\sqrt{n}\,\bar{\pmb{g}}\Rightarrow \mathcal{N}(\pmb{0},\pmb{S}_0)$, hence $\bar{\pmb{g}}=O_{\mathbb{P}}(n^{-1/2})$. Also, $\pmb{S}_n\to \pmb{S}_0$ in probability with $\pmb{S}_0$ positive definite, so $\pmb{S}_n$ is invertible w.h.p. Moreover, Assumption~\ref{ass:moments} implies $\max_{1\le i\le n}\|\pmb{g}_i\|_2=O_{\mathbb{P}}(1)$ because the observations are i.i.d.\ at the subject level and $m_i$ is bounded.

Next we establish a local expansion of $\pmb{\Psi}(\pmb{\lambda})$ around $\pmb{\lambda}=\pmb{0}$. For $\|\pmb{\lambda}\|_2$ sufficiently small, use the identity
\[
\frac{1}{1+\pmb{\lambda}^{\top}\pmb{g}_i}
=
1-\pmb{\lambda}^{\top}\pmb{g}_i + \frac{(\pmb{\lambda}^{\top}\pmb{g}_i)^2}{1+\pmb{\lambda}^{\top}\pmb{g}_i},
\]
which yields
\[
\pmb{\Psi}(\pmb{\lambda})
=
\bar{\pmb{g}}
-\pmb{S}_n\,\pmb{\lambda}
+\pmb{R}_n(\pmb{\lambda}),
\]
where
\[
\pmb{R}_n(\pmb{\lambda})
=
\frac{1}{n}\sum_{i=1}^{n}
\pmb{g}_i\frac{(\pmb{\lambda}^{\top}\pmb{g}_i)^2}{1+\pmb{\lambda}^{\top}\pmb{g}_i}.
\]
On the event $\max_i |\pmb{\lambda}^{\top}\pmb{g}_i|\le 1/2$ (which holds w.h.p. for $\|\pmb{\lambda}\|_2\le c n^{-1/2}$ and fixed $c$),
\[
\|\pmb{R}_n(\pmb{\lambda})\|_2
\le
\frac{2}{n}\sum_{i=1}^{n}\|\pmb{g}_i\|_2\,(\pmb{\lambda}^{\top}\pmb{g}_i)^2
\le
2\|\pmb{\lambda}\|_2^2\cdot \frac{1}{n}\sum_{i=1}^{n}\|\pmb{g}_i\|_2^3
=
O_{\mathbb{P}}(\|\pmb{\lambda}\|_2^2),
\]
and thus $\sup_{\|\pmb{\lambda}\|\le c n^{-1/2}}\|\pmb{R}_n(\pmb{\lambda})\|_2=o_{\mathbb{P}}(n^{-1/2})$.

Define the candidate approximation $\pmb{\lambda}^{\ast}=\pmb{S}_n^{-1}\bar{\pmb{g}}$. Then $\|\pmb{\lambda}^{\ast}\|_2=O_{\mathbb{P}}(n^{-1/2})$ and
\[
\pmb{\Psi}(\pmb{\lambda}^{\ast})
=
\bar{\pmb{g}}-\pmb{S}_n\pmb{\lambda}^{\ast}+\pmb{R}_n(\pmb{\lambda}^{\ast})
=
\pmb{R}_n(\pmb{\lambda}^{\ast})
=
o_{\mathbb{P}}(n^{-1/2}).
\]
Also, the Jacobian of $\pmb{\Psi}$ is
\[
\frac{\partial \pmb{\Psi}(\pmb{\lambda})}{\partial \pmb{\lambda}^{\top}}
=
-\frac{1}{n}\sum_{i=1}^{n}\frac{\pmb{g}_i\pmb{g}_i^{\top}}{\{1+\pmb{\lambda}^{\top}\pmb{g}_i\}^{2}},
\]
which is negative definite in a neighborhood of $\pmb{0}$ w.h.p. because $\pmb{S}_n$ is positive definite and the denominators stay bounded. Therefore, by the implicit function theorem / Newton-Kantorovich argument, there exists a unique root $\widehat{\pmb{\lambda}}$ of $\pmb{\Psi}(\pmb{\lambda})=\pmb{0}$ in the ball $\{\|\pmb{\lambda}\|\le c n^{-1/2}\}$ w.h.p., and it satisfies
\[
\widehat{\pmb{\lambda}}-\pmb{\lambda}^{\ast}
=
\left\{\frac{\partial \pmb{\Psi}(\tilde{\pmb{\lambda}})}{\partial \pmb{\lambda}^{\top}}\right\}^{-1}\pmb{\Psi}(\pmb{\lambda}^{\ast})
=
o_{\mathbb{P}}(n^{-1/2}),
\]
for some $\tilde{\pmb{\lambda}}$ between $\widehat{\pmb{\lambda}}$ and $\pmb{\lambda}^{\ast}$. 

This yields $\|\widehat{\pmb{\lambda}}\|_2=O_{\mathbb{P}}(n^{-1/2})$ and the expansion $\widehat{\pmb{\lambda}}=\pmb{S}_n^{-1}\bar{\pmb{g}}+o_{\mathbb{P}}(n^{-1/2})$.
\end{proof}

\begin{proof}[Proof of Lemma~\ref{lem:el_quad}]
Let $\pmb{g}_i=\pmb{g}_i(\pmb{\theta}_0)$, $\bar{\pmb{g}}=n^{-1}\sum_{i=1}^{n}\pmb{g}_i$ and $\pmb{S}_n=n^{-1}\sum_{i=1}^{n}\pmb{g}_i\pmb{g}_i^{\top}$. By Lemma~\ref{lem:U_linear}, $\bar{\pmb{g}}=O_{\mathbb{P}}(n^{-1/2})$ and $\pmb{S}_n\to \pmb{S}_0$ in probability with $\pmb{S}_0$ positive definite, so $\pmb{S}_n$ is invertible w.h.p.

We also need the feasibility (convex hull) condition ensuring the EL weights exist. Since $\mathbb{E}(\pmb{g}_i)=\pmb{0}$ at $\pmb{\theta}_0$ and $\mathrm{Var}(\pmb{g}_i)=\pmb{S}_0$ is positive definite, the origin lies in the interior of the convex hull of $\{\pmb{g}_i\}_{i=1}^{n}$ with probability tending to one; this is a standard EL fact under nondegeneracy and i.i.d.\ sampling \citep{Kolaczyk1994empirical,owen2001empirical}. Hence the Lagrange multiplier equation \eqref{eq:lambda_eq} admits a unique solution $\widehat{\pmb{\lambda}}$ w.h.p.

By Lemma~\ref{lem:lambda_expand}, $\widehat{\pmb{\lambda}}=\pmb{S}_n^{-1}\bar{\pmb{g}}+o_{\mathbb{P}}(n^{-1/2})$ and $\|\widehat{\pmb{\lambda}}\|_2=O_{\mathbb{P}}(n^{-1/2})$. Moreover, $\max_{1\le i\le n}\|\pmb{g}_i\|_2=O_{\mathbb{P}}(1)$ by Assumption~\ref{ass:moments} and bounded cluster size, so $\max_i |\widehat{\pmb{\lambda}}^{\top}\pmb{g}_i|=o_{\mathbb{P}}(1)$, which justifies the Taylor expansion of the log.

Using $\log(1+t)=t-\tfrac{1}{2}t^2+O(t^3)$ uniformly for $|t|\le o(1)$,
\begin{align*}
\ell(\pmb{\theta}_0)
&=
2\sum_{i=1}^{n}\log\{1+\widehat{\pmb{\lambda}}^{\top}\pmb{g}_i\} \\
&=
2\sum_{i=1}^{n}\left(\widehat{\pmb{\lambda}}^{\top}\pmb{g}_i-\frac{1}{2}(\widehat{\pmb{\lambda}}^{\top}\pmb{g}_i)^2\right)
+
2\sum_{i=1}^{n}O\!\left(|\widehat{\pmb{\lambda}}^{\top}\pmb{g}_i|^3\right).
\end{align*}
The cubic remainder is $o_{\mathbb{P}}(1)$ because $\sum_i |\widehat{\pmb{\lambda}}^{\top}\pmb{g}_i|^3 \le (\max_i |\widehat{\pmb{\lambda}}^{\top}\pmb{g}_i|)\sum_i (\widehat{\pmb{\lambda}}^{\top}\pmb{g}_i)^2 = o_{\mathbb{P}}(1)\cdot O_{\mathbb{P}}(1)$. Also, $\sum_{i=1}^{n}\widehat{\pmb{\lambda}}^{\top}\pmb{g}_i = n\,\widehat{\pmb{\lambda}}^{\top}\bar{\pmb{g}}$ and $\sum_{i=1}^{n}(\widehat{\pmb{\lambda}}^{\top}\pmb{g}_i)^2 = n\,\widehat{\pmb{\lambda}}^{\top}\pmb{S}_n\widehat{\pmb{\lambda}}$. Therefore,
\[
\ell(\pmb{\theta}_0)
=
2n\,\widehat{\pmb{\lambda}}^{\top}\bar{\pmb{g}}
-
n\,\widehat{\pmb{\lambda}}^{\top}\pmb{S}_n\widehat{\pmb{\lambda}}
+o_{\mathbb{P}}(1).
\]
Substituting $\widehat{\pmb{\lambda}}=\pmb{S}_n^{-1}\bar{\pmb{g}}+o_{\mathbb{P}}(n^{-1/2})$ yields
\[
\ell(\pmb{\theta}_0)
=
n\,\bar{\pmb{g}}^{\top}\pmb{S}_n^{-1}\bar{\pmb{g}}+o_{\mathbb{P}}(1),
\]
as claimed.
\end{proof}

\begin{proof}[Proof of Theorem~\ref{thm:wilks}]
By Lemma~\ref{lem:el_quad},
\[
\ell(\pmb{\theta}_0) = n\,\bar{\pmb{g}}(\pmb{\theta}_0)^{\top}\pmb{S}_n(\pmb{\theta}_0)^{-1}\bar{\pmb{g}}(\pmb{\theta}_0) +o_{\mathbb{P}}(1).
\]
Lemma~\ref{lem:U_linear} gives $\sqrt{n}\,\bar{\pmb{g}}(\pmb{\theta}_0)\Rightarrow \mathcal{N}(\pmb{0},\pmb{S}_0)$ and $\pmb{S}_n(\pmb{\theta}_0)\to \pmb{S}_0$ in probability. Let $\pmb{S}_0^{1/2}$ be the symmetric square root and define $\pmb{Z}_n=\pmb{S}_n(\pmb{\theta}_0)^{-1/2}\sqrt{n}\,\bar{\pmb{g}}(\pmb{\theta}_0)$. Then $\pmb{Z}_n\Rightarrow \mathcal{N}(\pmb{0},\pmb{I}_d)$ by Slutsky, hence
\[
n\,\bar{\pmb{g}}^{\top}\pmb{S}_n^{-1}\bar{\pmb{g}}
=
\pmb{Z}_n^{\top}\pmb{Z}_n
\Rightarrow
\chi^2_d,
\]
which proves $\ell(\pmb{\theta}_0)\Rightarrow \chi^2_d$.

For the profile statistic, partition $\pmb{\theta}=(\pmb{\theta}_1^{\top},\pmb{\theta}_2^{\top})^{\top}$ with $\dim(\pmb{\theta}_1)=r$. Near $\pmb{\theta}_0$, the minimizer $\widehat{\pmb{\theta}}_2(\pmb{\theta}_1)=\arg\min_{\pmb{\theta}_2}\ell(\pmb{\theta}_1,\pmb{\theta}_2)$ exists and is unique w.h.p. because $\ell(\pmb{\theta})$ is locally convex and twice differentiable in a neighborhood of $\pmb{\theta}_0$ under Assumptions~\ref{ass:link}--\ref{ass:H}, the convexity follows from the EL dual form and nondegeneracy of $\pmb{S}_0$, see \cite{owen2001empirical,Kolaczyk1994empirical}.

Applying the quadratic approximation in Lemma~\ref{lem:el_quad} to $\ell(\pmb{\theta})$ yields, uniformly locally,
\[
\ell(\pmb{\theta})
=
n\,\bar{\pmb{g}}(\pmb{\theta})^{\top}\pmb{S}_n(\pmb{\theta})^{-1}\bar{\pmb{g}}(\pmb{\theta})
+o_{\mathbb{P}}(1).
\]
Using the smoothness of $\bar{\pmb{g}}(\pmb{\theta})$ and a Taylor expansion around $\pmb{\theta}_0$,
\[
\sqrt{n}\,\bar{\pmb{g}}(\pmb{\theta})
=
\sqrt{n}\,\bar{\pmb{g}}(\pmb{\theta}_0)
+\pmb{G}\sqrt{n}(\pmb{\theta}-\pmb{\theta}_0)+o_{\mathbb{P}}(1),
\]
where $\pmb{G}=\partial \pmb{U}_0(\pmb{\theta})/\partial \pmb{\theta}^{\top}|_{\pmb{\theta}=\pmb{\theta}_0}=\pmb{H}_0$. 

Minimizing the resulting quadratic form over $\pmb{\theta}_2$ yields a reduced quadratic form in the $r$-dimensional component corresponding to $\pmb{\theta}_1$, which converges to $\chi^2_r$ by the same self-normalized CLT argument as above (formally, via the Schur complement of the partitioned information matrix). Therefore $\ell_{\mathrm{prof}}(\pmb{\theta}_{1,0})\Rightarrow \chi^2_r$.
\end{proof}

\end{document}